\documentclass{lmcs}
\pdfoutput=1

\usepackage{lastpage}
\lmcsdoi{16}{2}{9}
\lmcsheading{}{\pageref{LastPage}}{}{}%
{Oct.~05,~2019}{Jun.~03,~2020}{}

\keywords{Service Contracts,
Contract Automata,
Controller Synthesis,
Orchestration,
Choreography}

\usepackage[T1]{fontenc}
\usepackage[utf8]{inputenc}
\usepackage{url}

\usepackage{thm-restate}

\usepackage{amsmath,amssymb,amstext,latexsym}
\usepackage[misc]{ifsym}
\usepackage{stmaryrd}
\usepackage{wasysym}
\usepackage{eurosym}

\usepackage{pifont}
\newcommand{\cmark}{\text{\ding{51}}}
\newcommand{\xmark}{\text{\ding{55}}}

\usepackage{verbatim}

\usepackage[arrow,matrix,cmtip,frame,rotate,color]{xy}
\CompileMatrices%
\UseCrayolaColors%

\usepackage{xcolor}
\usepackage{graphicx}

\usepackage{tikz}
\usetikzlibrary{arrows,automata,shapes.geometric,positioning}

\usepackage{booktabs}

\newcommand{\prset}{\texttt{Pr}}

\newcommand{\Alice}{\textit{Alice}}
\newcommand{\Bob}{\textit{Bob}}
\newcommand{\Carol}{\textit{Carol}}

\usepackage{xifthen}
\newcommand{\inv}[1][]{%
\ifthenelse{\equal{#1}{}}{I}{I(#1)}
}
\newcommand{\con}[1][]{%
\ifthenelse{\equal{#1}{}}{\chi}{\chi(#1)}%
}

\newcommand{\acronym}{MSCA}
\newcommand{\permittedset}{A^{\Diamond}}
\newcommand{\necessaryset}{A^{\Box}}

\newcommand{\offerset}{A^o}
\newcommand{\Boxurgent}{\Box_{\textit{unc}}}

\newcommand{\Permitted}{\Diamond}

\newcommand{\Necessary}{\Box}

\newcommand{\alert}[1]{ \color{red} #1 \color{black} }

\newcommand{\ithel}[2]{{#1}_{({#2})}}
\newcommand{\ithelprime}[2]{{#1}\,'_{\!({#2})}}
\newcommand{\blk}{\bullet}
\newcommand{\Oset}{\mathsf{O}}
\newcommand{\Rset}{\mathsf{R}}
\newcommand{\Lset}{\mathsf{L}}

\newcommand{\TRANS}[1]{\xrightarrow{#1}}

\newcommand{\TRANSS}[1]{{\xrightarrow{\raisebox{-.3ex}[0pt][0pt]{\scriptsize $#1$}}}}

\renewcommand{\epsilon}{\varepsilon}

\newcommand{\hbra}{
\hbox to 1 \textwidth{\vrule width0.3mm height 1.8mm depth-0.3mm 
                    \leaders\hrule height1.8mm depth-1.5mm\hfill
                    \vrule width0.3mm height 1.8mm depth-0.3mm}}
\newcommand{\hket}{
\hbox to 1 \textwidth{\vrule width0.3mm height1.5mm 
                    \leaders\hrule height0.3mm\hfill
                    \vrule width0.3mm height1.5mm}}

\usepackage{etoolbox}
\newtoggle{APPENDIX}
\toggletrue{APPENDIX}  
\newcommand{\ignore}[1]{}

\newcommand{\modiff}[1]{\color{black}#1\color{black}}

\apptocmd{\sloppy}{\hbadness 10000\relax}{}{} 
\begin{document}

\title[Synthesis of Orchestrations and Choreographies]{Synthesis of Orchestrations and Choreographies: Bridging the Gap between Supervisory Control and Coordination of Services}

\author[D.~Basile]{Davide Basile\rsuper{a}} 
\author[M.H.~ter~Beek]{Maurice H. ter Beek\rsuper{a}} 
\address{\lsuper{a}ISTI--CNR, Pisa, Italy} 
\email{\{davide.basile,maurice.terbeek\}@isti.cnr.it}

\author[R.~Pugliese]{Rosario Pugliese\rsuper{b}} 
\address{\lsuper{b}University of Florence, Italy}
\email{rosario.pugliese@unifi.it}


\begin{abstract}\noindent
We present a number of contributions to bridging the gap between supervisory control theory and coordination of services in order to explore the frontiers between coordination and control systems.
Firstly, we modify the classical synthesis algorithm from supervisory control theory for obtaining the so-called most permissive controller in order to synthesise orchestrations and choreographies of service contracts formalised as contract automata.
The key ingredient to make this possible is a novel notion of controllability.
Then, we present an abstract parametric synthesis algorithm and show that it generalises the classical synthesis as well as the orchestration and choreography syntheses.
Finally, through the novel abstract synthesis, we show that the concrete syntheses are in a refinement order.
\modiff{A running example from the service domain illustrates our contributions.}
\end{abstract}

\maketitle


\section{Introduction}\label{sect:introduction}

\modiff{Services are ubiquitous in today's society. Examples include finances, healthcare, and tourism (e.g.\ booking services). Service-oriented computing (SOC) is \lq\lq the discipline that seeks to develop computational abstractions, architectures, techniques, and tools to support services broadly\rq\rq~\cite{manifesto}. According to this paradigm, services are well-defined, self-contained, and stand-alone software modules that provide some standard business functionality. As such, services can serve as building blocks for the rapid, low-cost development of distributed applications in heterogeneous environments. Services used in composite applications are not limited to new service implementations, but may also include adapted and wrapped fragments of existing applications. The strength of SOC is composing multiple, distributed services into more powerful applications. This reuse through composition provides businesses a means to reduce the cost and risks of developing new applications.}

\modiff{Service composition is thus a key challenge for the full realisation of the SOC paradigm. As such, it can benefit from and contribute to emerging research directions inspired by cloud computing, IoT, social computing,
and mobile computing, to name but a few.
For instance, the composition of cloud services requires the coordination of hardware and software resources across various layers.
The IoT concept of smart cities concerns the large-scale composition of diverse and heterogeneous digital devices and services to provide multiple real-time, end user customised functionalities.
Service composition based on relations in today's large social networks is challenging due to size and complexity of the resultant big data.
In mobile environments, service composition is required to consider the intrinsic dynamicity and its effect on QoS aspects concerning security and reliability.}

\modiff{Two approaches are widely adopted for coordinating services by means of service composition}: \emph{orchestration} and \emph{choreography}. Intuitively, an orchestration yields the description of a distributed workflow from \lq\lq one party's perspective\rq\rq~\cite{Pel03}, whereas a choreography describes the behaviour of the involved parties from a \lq\lq global viewpoint\rq\rq~\cite{WS-CDL}. In an orchestrated model, the service components are coordinated by a special component, the \emph{orchestrator}, which, by interacting with them, dictates the workflow at runtime. In a choreographed model, instead, the service components autonomously execute and interact with each other on the basis of a local control flow expected to comply with their role as specified by the global viewpoint. Ideally, a choreographed model is thought to be more efficient due to the absence of the overhead of communications with the orchestrator. Any choreography can trivially be transformed into an orchestration of services, by adding an idle orchestrator. Similarly, by explicitly adding an orchestrator and its interactions with the service components, and hence the relative overhead, an orchestration of services can be transformed into a choreography.

\modiff{Despite the key impact that SOC can have on other contemporary computing paradigms, as already mentioned before, the recent Service Computing Manifesto~\cite{manifesto} points out that \lq\lq Service systems have so far been built without an adequate rigorous foundation that would enable reasoning about them\rq\rq\ and that \lq\lq The design of service systems should build upon a formal model of services\rq\rq. Therefore, the principled design of service-based applications and systems is identified as a primary research challenge for the coming years.}

To tackle this challenge, in~\cite{BDFT15}, two orchestrated and choreographed automata-based models of services, called \emph{contract automata}\footnote{Not to be confused with the accidentally homonymous contract automata of~\cite{APSS16}, which were introduced to formalise legal contracts among two parties expressed in natural language.} and \emph{communicating (finite-state) machines}, respectively, are studied and related.
The goal of both formalisms is to compose the automata such that each service is capable of reaching an accepting (final) state by synchronous/asynchronous one-to-one interactions with the other services in the composition.
The main difference relies on the fact that contract automata are oblivious of their partners and an orchestration is synthesised to drive their interactions, whereas communicating machines name the recipient service of each interaction upfront and use FIFO buffers to interact with each other. The model of contract automata was further developed in~\cite{BDF16}.

The orchestration synthesis was borrowed from the synthesis of the most permissive controller (mpc) from Supervisory Control Theory (SCT)~\cite{RW87,CL06}, whose aim is to coordinate an ensemble of (local) components into a (global) system that functions correctly.
In the context of contract automata, this amounts to refining the composition of service contracts into its largest sub-portion whose behaviour is non-blocking and safe (a notion of service compliance).
The adaptation of the mpc synthesis for synthesising an orchestration of services required the introduction of a novel notion of \emph{semi-controllability}.
Basically, the assumption of the presence of an unpredictable environment was dropped in favour of a milder notion of  predictable necessary service requests to be fulfilled.

In this paper, \modiff{we contribute to the research efforts on rigorously modelling service orchestration and choreography.
More specifically}, building on~\cite{BBP19}, we report on the efforts to relate the mpc synthesis and the orchestration synthesis of contract automata through a polished, homogeneous formalisation.
The need for semi-controllability is showcased with intuitive examples and its expressiveness is evaluated with respect to standard SCT notions of controllable and uncontrollable actions.
Moreover, we introduce a novel choreography synthesis algorithm and a novel abstract synthesis algorithm. We then show that each of the three synthesis algorithms can be obtained through a different instantiation of this abstract synthesis algorithm.
This paper extends~\cite{BBP19} in several ways. We include all proofs and as an additional contribution we demonstrate that the different instantiations of the abstract synthesis algorithm are related through a notion of refinement, which allows us to formally prove that the orchestration synthesis is an abstraction of the mpc synthesis.
\modiff{Furthermore, we illustrate each of the synthesis algorithms through a running example from the service domain.
Finally, we have also extended the prototypical tool FMCAT\footnote{\modiff{FMCAT is available at \url{https://github.com/davidebasile/FMCAT}. A video-tutorial showcasing the specification, composition, and syntheses of the contract automata from our running example is available at \url{https://github.com/davidebasile/FMCAT/tree/master/demoLMCS2020}.}} with the implementation of the novel choreography synthesis algorithm and then used it to compute all the automata compositions and syntheses shown in our running example.}



The paper is organised as follows.
Section~\ref{sect:background} contains background notions and results concerning contract automata and SCT,
\modiff{and introduces our running example}.
Section~\ref{sect:orchestrationsynthesis} and Section~\ref{sect:choreographysynthesis} introduce the synthesis of orchestrations and the novel synthesis of choreographies in the setting of (modal service) contract automata.
Section~\ref{sect:discussion} demonstrates that each of the introduced synthesis algorithms is an instantiation of a more abstract, parametric synthesis algorithm, and Section~\ref{sect:po}\ shows that these different instantiations are related.
Section~\ref{sect:related} discusses related work, while Section~\ref{sect:conclusion} concludes the paper and provides some hints for future work.
\iftoggle{APPENDIX}
{\modiff{Appendix~\ref{app:proofs} contains the full proofs of two results only sketched in Section~\ref{sect:discussion}.}
}
{}

\section{Background}\label{sect:background}

In this section, we provide some background useful to better appreciate our contributions on the crossroads of supervisory control theory and coordination of services formalised as modal service contract automata. \modiff{We also introduce a running example from the service domain that will be used throughout the paper to illustrate our contributions.}

\subsection{Contract Automata}

A Contract Automaton (CA) represents either a single service (in which case it is called a \emph{principal}) or a multi-party composition of services performing actions. The number of principals of a CA is called its \emph{rank}.
The states of a CA are vectors of states of principals. In the following, $\vec v$ denotes a vector and $\ithel {\vec v} i$ denotes its $i$th element.

The transitions of CA are labelled with actions, which are vectors of elements in the finite set of \emph{basic actions} $\Lset = \Rset \cup \Oset \cup \{\blk\}$, with $\Rset \cap \Oset = \emptyset$ and $\blk \not\in \Rset \cup \Oset$.
Intuitively, $\Rset$ is the set of \emph{requests} (depicted as non-overlined labels on arcs, e.g.\ $a$), $\Oset$ is the set of \emph{offers} (depicted as overlined labels on arcs, e.g.\ $\overline{a}$) with \ $\Oset = \{\,\overline a\mid a \in \Rset\,\}$, and $\blk$ is a distinguished symbol representing the \emph{idle} action.
To establish if a pair of a request and an offer are \emph{complementary}, we use the \emph{involution} function $co: \Lset \rightarrow \Lset$ defined as follows:
$\forall a \in \Rset: co(a) = \overline a$, $\forall \overline a \in \Oset: co(\overline a) = a$,
and $co(\blk)=\blk$.
By abusing notation, we let $co(\Rset)=\Oset$ and $co(\Oset)=\Rset$.

An \emph{action} is a vector $\vec a$ of basic actions with either a single offer, or a single request, or a single pair of request-offer that match, i.e.\ there exist $i$ and $j$ such that $\ithel {\vec a} i$ is an offer and $\ithel {\vec a} j$ is the complementary request (formally $co(\ithel {\vec a} i)=\ithel {\vec a} j$); all other elements of the vector are~$\blk$, meaning that the corresponding principals remain idle.
Such action is called \emph{request}, \emph{offer}, or \emph{match}, respectively.
A transition is said to be a request, offer, or match according to its labelling action.

The goal of each principal is to reach an accepting (\emph{final}) state such that all its requests and offers are matched.

In~\cite{BDGDF17}, CA are equipped with \emph{modalities}, i.e.\ \emph{permitted} ($\Permitted$) and \emph{necessary} ($\Box$), that are associated to requests.
\modiff{Offers remain without modalities, i.e.\ they are interpreted as always permitted, like in the original CA formalism. Matches, on the other hand, inherit the modality of the involved request}.
The resulting formalism, called Modal Service Contract Automata (MSCA), is formally defined next. Differently from standard SCT, all transitions of MSCA are \emph{observable}, since MSCA model the execution of services in terms of their requests and offers.

\begin{defi}[MSCA~\cite{BDGDF17}]\label{def:contract}
Given a finite set of states $\mathcal{Q} = \{q_1,q_2, \ldots \}$, a \emph{Modal Service Contract Automata (MSCA)} $\mathcal{A}$ of rank $n$ is a septuple $\langle Q, \vec{q_0}, A^{\Permitted}, A^{\Box}, A^{o}, T, F \rangle$,
with set of states $Q  \subseteq \mathcal{Q}^n$,
initial state $\vec{q_0} \in Q$,
%
set of permitted requests $\permittedset$ and of necessary request $\necessaryset$ partitioning the set of requests $A^{r} \subseteq \Rset$,
%
%
set of offers $A^{o} \subseteq \Oset$, set of final states $F \subseteq Q$,
set of transitions $T \subseteq Q \times A  \times Q$, where
$A\subseteq{(A^r \cup \offerset \cup \{\blk\})}^n$,
partitioned into \emph{permitted} transitions $T^\Diamond$ and \emph{necessary} transitions $T^\Box$,
such that:
(i)~given $t= (\vec{q},\vec{a},\vec{q}\,') \in T$, $\vec{a}$ is either a request, or an offer, or a match;
(ii)~$\forall i \in 1\ldots n,\ \ithel{\vec{a}} i=\blk$ implies $\ithel{\vec{q}} i=\ithelprime{\vec{q}} i$;
(iii)~$t \in T^\Diamond$ if and only if $\vec a$ is either a request, or a match on $a \in A^\Diamond$, or an offer on $\overline a \in \offerset$; otherwise $t \in T^\Box$.
\end{defi}

Remarkably, it follows that the set of transitions of an MSCA is finite.

A \emph{principal} is an MSCA of rank 1 such that $A^r \cap co(A^o) = \emptyset$.
Unless stated differently, we assume that it is given an MSCA
$\mathcal A = \langle Q_\mathcal{A}, \vec{q_0}_\mathcal{A}, A^{\Permitted}_{\mathcal{A}}, A^{\Box}_{\mathcal{A}}, A^{o}_{\mathcal{A}}, T_{\mathcal{A}}, F_{\mathcal{A}} \rangle$ of rank $n$. Subscript $\mathcal A$ may be omitted if no confusion may arise.

A \emph{step} $(w,\vec{q}\,) \TRANSS{\vec a} (w',\vec{q}\,')$ occurs in $\mathcal A$ if and only if $w=\vec{a} w'$, $w'\in A^*$, and $(\vec{q},\vec{a},\vec{q}\,')\in T$.
Let $\rightarrow^*$ be the reflexive and transitive closure of $\rightarrow$.
The \emph{language} of $\mathcal{A}$ is $\mathcal{L}(\mathcal{A})= \{\, w  \mid  (w,\vec{q_0}) \TRANSS{w}^* (\epsilon, \vec{q}\,),\ \vec{q} \in F \,\}$.
A step may be denoted as $\vec{q}\, \TRANSS{\vec a}$ if $w$, $w'$, and $\vec{q}\,'$ are irrelevant, and
as $\vec{q}\,\rightarrow \vec{q}\,'$ if $w$, $w'$, and $\vec a$ are irrelevant.

Composition of services is rendered through the composition of their MSCA models.
This amounts to interleaving or matching the transitions of the component MSCA, forcing the match whenever two components are ready on their respective complementary request/offer actions.
In the resulting MSCA, states and actions are vectors of states and actions of the component MSCA, respectively.
The composition is non-associative, i.e.\ pre-existing matches are not rearranged if a new MSCA joins the composition afterwards.

In a composition of MSCA, typically various properties are analysed. We are especially interested in \emph{agreement} and \emph{strong agreement} (which in the literature is also known as progress of interactions, deadlock freedom, compliance or conformance of contracts).
In an MSCA in strong agreement, all requests and offers must be matched.
Instead, the property of agreement only requires matching all requests.
An MSCA admits (strong) agreement if it has a trace satisfying the corresponding property, and it is \emph{safe} if all its traces are such.

\modiff{The MSCA formalism has its origins in~\cite{BDF16}, where CA were first introduced, but in this paper we build on the version with modalities from~\cite{BDGDF17} to cater for controllable and uncontrollable (and thus semi-controllable) actions.
The branching condition for CA from~\cite{BDFT15} will be recalled in Section~\ref{sect:choreographysynthesis} as a condition for obtaining a choreography from an orchestration, and it is satisfied by construction by the output MSCA of the synthesis of the choreography.}


\begin{exa}
\modiff{We introduce a running example that will be used throughout the paper to showcase the synthesis of orchestration and choreography.
We anticipate, as discussed in detail in Section~\ref{sect:choreographysynthesis}, that a modified version of \acronym\ is used for the synthesis of a choreography, in which offers can be necessary whilst requests are only permitted.}

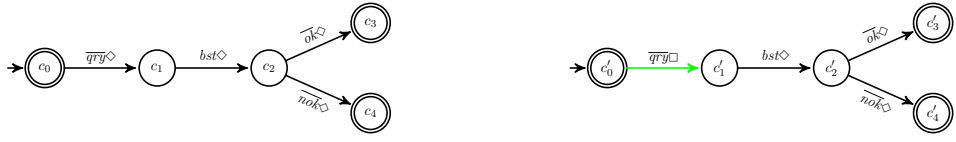
\begin{figure}
\centering
\begin{tikzpicture}[->,scale=.5,>=stealth',shorten >=1pt,auto,node distance=1cm,
        semithick, every node/.style={scale=0.5},initial text={},inner sep=0pt, minimum size=0pt]
\tikzstyle{every state}=[fill=white,draw=black,text=black]

\node[initial,state,accepting] (0) {$c_{0}$};
\node[state] (1) [right = of 0] {$c_{1}$};
\node[state] (2) [right = of 1] {$c_{2}$};
\node[state,accepting] (3) [above right = .25cm and 1cm of 2] {$c_{3}$};
\node[state,accepting] (4) [below right = .25cm and 1cm of 2] {$c_{4}$};

\node[initial,state,accepting] (00) [right = 4cm of 2] {$c'_{0}$};
\node[state] (11) [right = of 00] {$c'_{1}$};
\node[state] (22) [right = of 11] {$c'_{2}$};
\node[state,accepting] (33) [above right = .25cm and 1cm of 22] {$c'_{3}$};
\node[state,accepting] (44) [below right = .25cm and 1cm of 22] {$c'_{4}$};

\path
(0) edge[] node[above,yshift=3pt]{$\overline{\textit{qry}}\Permitted$} (1)
(1) edge[] node[above,yshift=5pt]{$\textit{bst}\Permitted$} (2)
(2) edge[] node[sloped,anchor=center,yshift=10pt]{$\overline{\textit{ok}}\Permitted$} (3)
(2) edge[] node[sloped,anchor=center,yshift=-10pt]{$\overline{\textit{nok}}\Permitted$} (4)

(00) edge[draw=green] node[above,yshift=3pt]{$\overline{\textit{qry}}\Necessary$} (11)
(11) edge[] node[above,yshift=5pt]{$\textit{bst}\Permitted$} (22)
(22) edge[] node[sloped,anchor=center,yshift=10pt]{$\overline{\textit{ok}}\Permitted$} (33)
(22) edge[] node[sloped,anchor=center,yshift=-10pt]{$\overline{\textit{nok}}\Permitted$} (44);
\end{tikzpicture}
\caption{\label{fig:client}\color{black}{MSCA $\textsf{Client}$ (left) and $\textsf{PrivilegedClient}$ (right)}}
\end{figure}

\modiff{Figures~\ref{fig:client},~\ref{fig:broker}, and~\ref{fig:hotel} show five \acronym\ of rank 1.
These automata model an example of a hotel booking service, where clients and hotels interact by means of a broker for booking hotel rooms.
There are two types of clients, \texttt{Client} and \texttt{PrivilegedClient}.
Both clients can either terminate without interactions (final states are drawn as double circles), or they can engage in interactions with the broker to possibly book a room.
The first interaction is to ask for a room, by means of the offer $\overline{\textit{qry}}$ (query).
After this action, the clients receive the best room option from the broker, by means of the request $bst$ (best).
Then, each client can either decide to accept (offer $\overline{\textit{ok}}$) or refuse (offer $\overline{nok}$) the option offered by the broker.
\texttt{PrivilegedClient} will be used to showcase a choreography in Section~\ref{sect:choreographysynthesis}.
Accordingly, \texttt{PrivilegedClient} only differs from \texttt{Client} with respect to the first offer $\overline{\textit{qry}}$, which is declared \emph{necessary}.
Basically, \texttt{PrivilegedClient} reaches an agreement only if there exists a trace in which its offer is necessarily matched.
All other actions are permitted.}

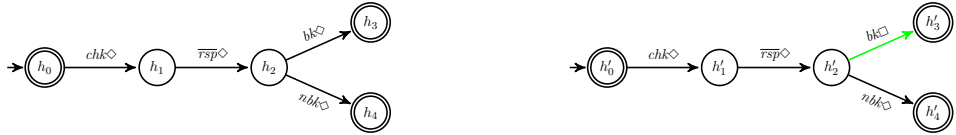
\begin{figure}
\centering
\begin{tikzpicture}[->,scale=.5,>=stealth',shorten >=1pt,auto,node distance=1cm,
        semithick, every node/.style={scale=0.5},initial text={},inner sep=0pt, minimum size=0pt]
\tikzstyle{every state}=[fill=white,draw=black,text=black]

\node[initial,state,accepting] (0) {$h_{0}$};
\node[state] (1) [right = of 0] {$h_{1}$};
\node[state] (2) [right = of 1] {$h_{2}$};
\node[state,accepting] (3) [above right = .25cm and 1cm of 2] {$h_{3}$};
\node[state,accepting] (4) [below right = .25cm and 1cm of 2] {$h_{4}$};

\node[initial,state,accepting] (00) [right = 4cm of 2] {$h'_{0}$};
\node[state] (11) [right = of 00] {$h'_{1}$};
\node[state] (22) [right = of 11] {$h'_{2}$};
\node[state,accepting] (33) [above right = .25cm and 1cm of 22] {$h'_{3}$};
\node[state,accepting] (44) [below right = .25cm and 1cm of 22] {$h'_{4}$};

\path
(0) edge[] node[above,yshift=5pt]{$\textit{chk}\Permitted$} (1)
(1) edge[] node[above,yshift=5pt]{$\overline{\textit{rsp}}\Permitted$} (2)
(2) edge[] node[sloped,anchor=center,yshift=10pt]{$\textit{bk}\Permitted$} (3)
(2) edge[] node[sloped,anchor=center,yshift=-10pt]{$\textit{nbk}\Permitted$} (4)

(00) edge[] node[above,yshift=5pt]{$\textit{chk}\Permitted$} (11)
(11) edge[] node[above,yshift=5pt]{$\overline{\textit{rsp}}\Permitted$} (22)
(22) edge[draw=green] node[sloped,anchor=center,yshift=10pt]{$\textit{bk}\Necessary$} (33)
(22) edge[] node[sloped,anchor=center,yshift=-10pt]{$\textit{nbk}\Permitted$} (44);
\end{tikzpicture}
\caption{\label{fig:hotel}\color{black}{MSCA $\textsf{Hotel}$ (left) and $\textsf{PrivilegedHotel}$ (right)}}
\end{figure}

\modiff{Similarly, there are two types of hotels, \texttt{Hotel} and \texttt{PrivilegedHotel}.
Also both hotels can either terminate without interactions, or they can engage in interactions with the broker to possibly have their rooms booked.
The first interaction is to receive a request for a room, by means of the request $\textit{chk}$ (check).
After this check, a response is sent to the broker through the offer $\overline{\textit{rsp}}$ (response).
Then, each hotel can either receive a booking or a no booking reply by means of requests $\textit{bk}$ (book) or $\textit{nbk}$ (no book), respectively.
\texttt{PrivilegedHotel} will be used to showcase an orchestration in Section~\ref{sect:orchestrationsynthesis}.
Accordingly, \texttt{PrivilegedHotel} only differs from \texttt{Hotel} with respect to the request $\textit{bk}$, declared \emph{necessary}.
Basically, \texttt{PrivilegedHotel} admits non-empty orchestrations only if there \emph{exists} a trace in which one of its rooms is booked (i.e.\ the necessary request is matched).
All other actions are permitted.}

\begin{figure}
\centering
\begin{tikzpicture}[->,scale=.5,>=stealth',shorten >=1pt,auto,node distance=1cm,
        semithick, every node/.style={scale=0.5},initial text={},inner sep=0pt, minimum size=0pt]
\tikzstyle{every state}=[fill=white,draw=black,text=black]

\node[initial,state,accepting] (0) {$b_{0}$};
\node[state] (1) [right = of 0] {$b_{1}$};
\node[state] (2) [right = of 1] {$b_{2}$};
\node[state] (3) [right = of 2] {$b_{3}$};
\node[state] (4) [right = of 3] {$b_{4}$};
\node[state] (5) [right = of 4] {$b_{5}$};
\node[state] (6) [right = of 5] {$b_{6}$};
\node[state] (7) [above right = .25cm and 1cm of 6] {$b_{7}$};
\node[state] (8) [right = of 7] {$b_{8}$};
\node[state,accepting] (9) [right = of 8] {$b_{9}$};
\node[state] (10) [below right = .25cm and 1cm of 6] {$b_{10}$};
\node[state] (11) [right = of 10] {$b_{11}$};
\node[state,accepting] (12) [right = of 11] {$b_{12}$};

\path
(0) edge[] node[above,yshift=5pt]{$\textit{qry}\Permitted$} (1)
(1) edge[] node[above,yshift=5pt]{$\overline{\textit{chk}}\Permitted$} (2)
(2) edge[] node[above,yshift=5pt]{$\textit{rsp}\Permitted$} (3)
(3) edge[] node[above,yshift=5pt]{$\overline{\textit{chk}}\Permitted$} (4)
(4.north east) edge[] node[above,yshift=5pt]{$\textit{rsp}\Permitted$} (5.north west)
(5.south west) edge[] node[below,yshift=-5pt]{$\overline{\textit{chk}}\Permitted$} (4.south east)
(5) edge[] node[above,yshift=5pt]{$\overline{\textit{bst}}\Permitted$} (6)
(6) edge[] node[sloped,anchor=center,yshift=10pt]{$\textit{ok}\Permitted$} (7)
(6) edge[] node[sloped,anchor=center,yshift=-10pt]{$\textit{nok}\Permitted$} (10)
(7) edge[] node[above,yshift=5pt]{$\overline{\textit{bk}}\Permitted$} (8)
(8) edge[] node[above,yshift=5pt]{$\overline{\textit{nbk}}\Permitted$} (9)
(9) edge[loop,looseness=7.5,in=315,out=45] node[right,xshift=3pt]{$\overline{\textit{nbk}}\Permitted$} (9)
(10) edge[] node[above,yshift=5pt]{$\overline{\textit{nbk}}\Permitted$} (11)
(11) edge[] node[above,yshift=5pt]{$\overline{\textit{nbk}}\Permitted$} (12)
(12) edge[loop,looseness=7.5,in=315,out=45] node[right,xshift=3pt]{$\overline{\textit{nbk}}\Permitted$} (12);
\end{tikzpicture}
\caption{\label{fig:broker}\modiff{MSCA $\textsf{Broker}$}}
\end{figure}

\modiff{Finally, the \texttt{Broker} acts as an intermediary between a client and at least two hotels. The broker starts by receiving a request for a room by a client through the request $\textit{qry}$. At this point, it starts to interact with the hotels to search for a possible option to propose to the client. This is done by (twice) repeating the offer $\overline{\textit{chk}}$ (sending a room enquiry) followed by the request $\textit{rsp}$ (receiving the room response by one hotel).
Indeed, at least two hotels must be enquired to speak of a best offer. After that, the \texttt{Broker} can engage with further hotels, from state $b_5$, or it can proceed with the best offer $\overline{\textit{bst}}$ to the client. At this point, it receives through the requests $\textit{ok}$ or $\textit{nok}$ either the acceptance or rejection, respectively, of its offer. If the offer is accepted, \texttt{Broker} proceeds to book the room with offer $\overline{\textit{bk}}$ to the selected hotel (abstracted away in the contract) and replying to all other hotels with a $\overline{\textit{nbk}}$ offer.
Otherwise, if the offer is rejected, \texttt{Broker} sends to all hotels waiting for a reply the offer $\overline{\textit{nbk}}$. All actions of \texttt{Broker} are permitted.}
\end{exa}


\subsection{Supervisory Control Theory}

The aim of Supervisory Control Theory~\cite{RW87,CL06} (SCT) is to provide an algorithm to synthesise a finite-state automaton model of a \emph{supervisory controller} from given (component) finite-state automata models of the uncontrolled system and its requirements, themselves expressed as automata.
The synthesised supervisory controller, if successfully generated, is such that the controlled system, which is the composition (i.e.\ synchronous product) of the uncontrolled system and the supervisory controller, satisfies the requirements and is additionally \emph{non-blocking}, \emph{controllable}, and \emph{maximally permissive}.

An automaton is \emph{non-blocking} if from each state at least one of the so-called \emph{marked states} (distinguished stable states representing completed \lq tasks\rq~\cite{RW87}, e.g.\ a final state)  can be reached without passing through so-called \emph{forbidden states}, meaning that the system always has the possibility to return to an accepted stable state. The algorithm assumes that marked states and forbidden states are indicated for each component model.

SCT distinguishes between \emph{observable} and \emph{unobservable}, as well as \emph{controllable} and \emph{uncontrollable} actions, where unobservable actions are also uncontrollable. Intuitively, the supervisory controller cannot distinguish one unobservable action from the other, whereas it can take observable actions apart.
Moreover, it is not permitted to directly block uncontrollable actions from occurring; the controller is only allowed to disable them by preventing controllable actions from occurring. Intuitively, controllable actions correspond to stimulating the system, while uncontrollable actions correspond to messages provided by the environment, like sensors, which may be neglected but cannot be denied from existing.

Finally, the fact that the resulting supervisory controller is \emph{maximally permissive} (or least restrictive) means that as much behaviour of the uncontrolled system as possible remains present in the controlled system without violating neither the requirements, nor controllability, nor the non-blocking condition.

From the seminal work of Ramadge and Wonham~\cite{RW87}, we know that a unique maximally permissive supervisory controller exists, provided that all actions are observable. This is called the \emph{most permissive controller} (\emph{mpc}); it coordinates an ensemble of (local) components into a (global) system that works correctly. The synthesis algorithm suffers from the same state space explosion problem as model checking~\cite{GW00}.
\modiff{However, SCT has successfully been applied to industrial size case studies~\cite{FMSR12,TBR12}}.

Intuitively, the synthesis algorithm for computing the mpc of a finite-state automaton $\mathcal A$ works as follows.
The mpc is computed through an iterative procedure that at each step~$i$ updates incrementally a set of states $R_{i}$ containing the \emph{bad} states, i.e.\ those states that cannot prevent a forbidden state to be eventually reached, and refines an automaton $\mathcal K_{i}$.

The algorithm starts with an automaton $\mathcal K_{0}$ equal to $\mathcal A$ and a set $R_{0}$ containing all \emph{dangling} states in $\mathcal A$,
where a state is dangling if it cannot be reached from the initial state or cannot reach a final state.
At each step $i$, the algorithm prunes from $\mathcal K_{i-1}$ in a backwards fashion transitions with target state in $R_{i-1}$ or forbidden source state.
The set $R_{i}$ is obtained by adding to $R_{i-1}$ dangling states in $\mathcal K_{i}$ and source states of uncontrollable transitions of $\mathcal A$ with target state in $R_{i-1}$.
When no more updates are possible, the algorithm terminates.
Termination is ensured since $\mathcal A$ is finite-state and has a finite set of transitions, and at each step the subsets of its states $R_{i}$ cannot decrease while the set of its transitions $T_{\mathcal K_i}$ cannot increase.
Now, suppose that at its termination the algorithm returns the pair $(\mathcal K_{s}, R_{s})$.
We have that the mpc is empty, if the initial state of $\mathcal A$ is in $R_{s}$; otherwise, the mpc is obtained from $\mathcal K_{s}$ by removing the states $R_{s}$.

We report below the standard synthesis algorithm, but we homogenise the notation and simplify the formulation, to align the algorithm with those presented in the next sections.
For this purpose, we assume the standard mpc synthesis to operate on MSCA where necessary transitions ($T^\Box$) are uncontrollable whilst permitted transitions ($T^\Permitted$) are controllable.

We use $\langle~\rangle$ to denote the empty automaton.
A state $q \in Q$ is said to be \emph{dangling} if and only if $\nexists\,w$ such that $q_0 \TRANSS{w}^*q$ or $q \TRANSS{w}^*q_f \in F$. Let $\textit{Dangling}(\mathcal{A})$ denote the set of dangling states of $\mathcal{A}$.
Given two MSCA $\mathcal A$ and $\mathcal A'$, we say that $\mathcal A'$ is a \emph{sub-automaton} of $\mathcal A$, denoted by $\mathcal A' \subseteq \mathcal A$, whenever the components of $\mathcal A'$ are included in the corresponding ones of $\mathcal A$.
Moreover, given
two sets of states $R$ and $R'$, we let $(\mathcal{A}, R) \leq (\mathcal{A}', R')$ if $\mathcal A' \subseteq \mathcal A$ and $R \subseteq R'$. It is straightforward to show that $(\textit{MSCA} \times 2^Q, \leq)$ is a complete partial order (cpo).

The algorithm to compute the mpc is now defined in terms of the \emph{least fixed point} of a monotone function on the cpo $(\textit{MSCA} \times 2^Q, \leq)$.

\begin{defi}[Standard synthesis, adapted from~\cite{RW87}]\label{def:mpc}
Let $\mathcal{A}$ be an MSCA, and let $\mathcal K_{0} = \mathcal A$ and $R_{0} = \textit{Dangling}(\mathcal{K}_{0})$.
We let the \emph{synthesis function}
$f: \textit{MSCA} \times  2^Q \rightarrow \textit{MSCA} \times  2^Q$ be defined as follows:
\begin{center}
\def\arraystretch{1.2}\begin{tabular}{c@{\hskip 0.5in}r@{\hskip 0.05in}c@{\hskip 0.05in}l}
\multicolumn{4}{l}{
$f(\mathcal{K}_{i-1},R_{i-1}) = (\mathcal{K}_{i},R_{i}),
 \text{ with}$}\\
& $T_{\mathcal{K}_{i}}$ & = & $T_{\mathcal{K}_{i-1}} \setminus \{\, (\vec{q} \TRANS{}\vec{q}\,') \in T_{\mathcal{K}_{i-1}} \mid \vec q\,' \in R_{i-1}\vee \vec q \text { is forbidden} \,\}$\\
& $R_{i}$ & = & $R_{i-1} \cup \{\, \vec q \mid (\vec{q} \TRANS{}\vec{q}\,') \in T_{\mathcal{A}}^\Box,\  \vec q\,'\in R_{i-1} \,\} \cup \textit{Dangling}(\mathcal{K}_{i})$
\end{tabular}
\end{center}
\end{defi}

\begin{thm}[Standard mpc, adapted from~\cite{RW87}]%
\label{the:mpc}
The synthesis function $f$ is monotone on the cpo $(\textit{MSCA} \times 2^Q, \leq)$ and its \emph{least fixed point} is:
\[
(\mathcal K_{s}, R_{s}) = \sup (\{\, f^n(\mathcal K_{0}, R_{0})\mid n \in \mathbb{N} \,\})
\]

\noindent
The \emph{mpc} of $\mathcal A$, denoted by $\mathcal K_{\mathcal A}$, is:
\[
\mathcal K_{\mathcal{A}} =
\left\{\begin{array}{l@{\qquad}l}
    \langle~\rangle
        & \mbox{if } \vec q_0 \in R_{s}\\
    \langle Q \setminus
    R_{s}, \vec q_0, A^{\Permitted}, A^{\Box}, A^{o},
        T_{\mathcal K_{s}}, F \setminus
        R_{s} \rangle
        & \mbox{otherwise}
\end{array}\right.
\]
\end{thm}

\modiff{We now want to estimate an upper bound of the complexity of the mpc synthesis algorithm as results from Definition~\ref{def:mpc} and Theorem~\ref{the:mpc}.
In the worst case, deciding if a state is dangling requires to visit the whole state space. Thus, an upper bound of the complexity of the procedure for deciding if a state is dangling is $\mathcal{O}(|Q|)$, and the upper-bound complexity for computing the set of dangling states is $\mathcal{O}(|Q|^2)$.
At each iteration, in the worst case, the algorithm either removes a single transition from~$T$ or adds a single state to~$R$, and each iteration requires to compute the set of dangling states.
Thus, an upper bound of the complexity of the mpc synthesis algorithm is $\mathcal{O}((|T|+|Q|)\times|Q|^2)$.
To conclude, it is worth noticing that our analysis focusses on the abstract specification of the algorithm while its implementation could be optimised, for example by using parallelism and dedicated data structures, in order to perform better than the complexity sketched above}.


\begin{exa}\label{ex:mpc}
\modiff{We continue the running example by discussing the synthesis of the mpc for the composition of two clients, the broker, one normal hotel, and one privileged hotel, denoted as
\[
\mathcal A_1 = \textsf{Client}\otimes \textsf{Client} \otimes \textsf{Broker}\otimes \textsf{Hotel} \otimes \textsf{PrivilegedHotel}
\]
The property to be enforced is agreement: each request must be matched by a corresponding offer. Basically, this property is an invariant stating that all request transitions are forbidden.
Since the synthesis 
works on forbidden states, we need to preprocess $\mathcal A_1$ accordingly.
In particular, the algorithm starts from the automaton $\mathcal A_1$ from which all permitted requests have been removed. Forbidden states are those featuring an outgoing \emph{necessary} request.}

\modiff{The resulting mpc only consists of the initial (and final) state $(c_0,c'_0,b_0,h_0,h'_0)$, and its behaviour is empty. Hence, agreement cannot be enforced in $\mathcal A_1$ using the standard synthesis algorithm.
This is an indication of the fact that standard mpc synthesis is not useful for the scope of synthesising a correct service composition (i.e.\ in which agreement is satisfied).
The reason is that necessary transitions are not to be interpreted as uncontrollable.
The notion of uncontrollable transition stems from the necessity of modelling an unpredictable environment, which is not suitable to model necessary service requests.
Basically, \texttt{PrivilegedHotel} has a necessary request that should be matched in \emph{at least} one trace of the composition. However, by interpreting such necessary request as uncontrollable, the synthesis is enforcing the necessary requests to be satisfied in \emph{every} trace of the composition. Intuitively, this would require that a client is not allowed to refuse to book a room.}

\modiff{As will become clear in the forthcoming sections, $\mathcal A_1$ admits a non-empty orchestration in which agreement is enforced, because necessary transitions will not be interpreted as fully uncontrollable.}

\modiff{We have used our tool FMCAT to calculate the automaton $\mathcal A_1$ and its mpc synthesis. Their computation time and state-space dimension are reported in Table~\ref{tab:results} (on page~\pageref{tab:results}).}
\end{exa}



\section{Synthesis of Orchestrations}\label{sect:orchestrationsynthesis}

In this section, we discuss how we revised the classical synthesis algorithm from SCT to obtain the mpc (cf.\ Theorem~\ref{the:mpc})  and synthesise orchestrations of MSCA\@.

Originally, MSCA were capable of expressing only permitted requirements, corresponding to actions that are controllable by the orchestrator. Hence, in the synthesis of the orchestration, all transitions labelled by actions violating the property to be enforced were pruned, and all dangling states were removed (cf.~\cite{BDF16}).

While permitted requests of MSCA are in one-to-one correspondence with controllable actions, interestingly this is not the case for necessary requests and uncontrollable actions.
A necessary (request) action is indeed a weaker constraint than an uncontrollable one.
This stems from the fact that traditionally uncontrollable actions relate to an unpredictable environment.
However, the interpretation of such actions as \emph{necessary} service requests to be fulfilled in a service contract, as is the case in the setting of MSCA, implies that it suffices that in the synthesised orchestration at least one such synchronisation (i.e.\ match) actually occurs.
This is precisely what is modelled by the notion of \emph{semi-controllable} actions, anticipated in~\cite{BDGDF17} and formally introduced in~\cite{JSCP,BBL19}, discussed next. 

The importance of this novel notion in the synthesis algorithm is showcased by an intuitive example.
Consider the two MSCA interacting on the necessary service request $a$ depicted in Fig.~\ref{fig:comp}~(left and middle), and their possible composition $\mathcal A$ depicted in Fig.~\ref{fig:comp}~(right). Note that $\mathcal A$ models two possibilities of fulfilling request $a$ from the leftmost automaton by matching it with a service offer $\overline{a}$ from the middle one. Note that a similar composition can  be obtained in other automata-based formalisms (such as, e.g., (timed) I/O automata~\cite{LT89,AD94,DLLNW10}).
Now assume that $a$ must be matched with $\overline{a}$ to obtain an agreement (i.e.\ it is \emph{necessary}), and that for some reason the \emph{bad} state {\color{red}$\xmark$} is to be avoided in favour of the \emph{successful} state {\color{green!60!black}$\cmark$}, i.e.\ in some sense we would like to express that $a$ must be matched at some point, rather than always.
In most automata-based formalisms
this is not allowed  and the resulting mpc is empty.
In the MSCA formalism,  
it is possible to orchestrate the composition of the two automata on the left in such a way that the result is the automaton $\mathcal A$ on the right, but \textit{without the state} {\color{red}$\xmark$} and its incident transition.

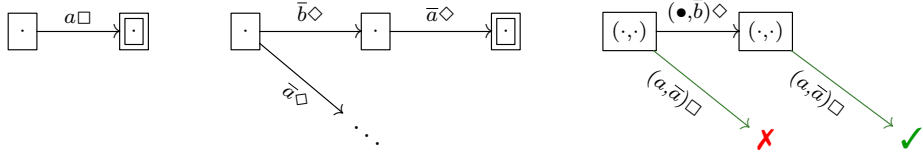
\begin{figure}
\centering\begin{tabular}{c}
\def\objectstyle{\scriptstyle}\def\labelstyle{\scriptstyle}
\xymatrix@C=.55cm@R=.8cm{
*+[F] {{}^{{}^{}}\cdot{}_{{}_{}}} \ar[rr]^{a\Box} & *{} & *+[F=] {{}^{{}^{}}\cdot{}_{{}_{}}} & *{} & *+[F] {{}^{{}^{}}\cdot{}_{{}_{}}} \ar[drr]_*+<-.4em>[@!-37]{\overline{a}\Permitted} \ar[rr]^{\overline{b}\Permitted} & *{} & *+[F] {{}^{{}^{}}\cdot{}_{{}_{}}} \ar[rr]^{\overline{a}\Permitted} & *{} & *+[F=] {{}^{{}^{}}\cdot{}_{{}_{}}} & *{} & *+[F] {(\cdot,\cdot)} \ar@[OliveGreen][drr]_*+<-.8em>[@!-37]{(a,\overline{a})\Box} \ar[rr]^{(\blk,\overline{b})\Permitted} & *{} & *+[F] {(\cdot,\cdot)} \ar@[OliveGreen][drr]_*+<-.8em>[@!-37]{(a,\overline{a})\Box} & {}\\
*{} & *{} & *{} & *{} & *{} & *{} & {\ddots\phantom{-}} & *{} & *{} & *{} & *{} & *{} & {\color{red}\text{\normalsize$\xmark$}} & *{} &  {\color{green!60!black}\text{\normalsize$\cmark$}}
}
\end{tabular}
\caption{\label{fig:comp}Two MSCA (left and middle) and a possible composition $\mathcal A$ of them (right)}
\end{figure}

In fact, in the MSCA formalism,  
$\mathcal A$ depicts a composition in which the automata on the left can synchronise on a so-called semi-controllable action $a\Box$ either in their initial state or after the middle automaton has performed some other action $\overline b\Permitted$, ignoring in this case whether a bad or a successful state is reached in the end.
Indeed, the notion of semi-controllability is independent from both the specific formalism being used and the requirement (e.g.\ agreement in case of MSCA) to be enforced.

As far as we know, we were the first to define a synthesis algorithm, in~\cite{BBL19}, that is capable of producing a controller that guarantees that \emph{at least} one of these two synchronisations actually occurs. 
Indeed, in the standard synthesis algorithm (cf.\ Theorem~\ref{the:mpc}),
action $a$ can either be \emph{controllable} and hence not necessary as we want, or \emph{uncontrollable} thus requiring that $a$ must \emph{always} be matched, a stronger requirement than the one posed by declaring $a$ as necessary.


To formalise the intuitions above\footnote{We refer the interested reader to~\cite{JSCP,BBL19} for more complete accounts.}, a semi-controllable transition $t$ becomes controllable if in a given portion of $\mathcal{A}$ there exists a semi-controllable match transition $t'$, with source and target states not dangling, such that in both $t$ and $t'$ the \emph{same} principal, in the \emph{same} local state, does the \emph{same} request. Otherwise, $t$ is uncontrollable. 

\begin{defi}[Controllability]\label{def:controllabilityorchestration}
Let $\mathcal A$ be an MSCA and let $t = (\vec q_1, \vec a_1, \vec q_1\!') \in T_{\mathcal A}$. Then:
\begin{itemize}
\item if $\vec a_1$ is an action on $a\in\permittedset \cup A^o$, then $t$ is \emph{controllable} (in $\mathcal A$) and part of $T^\Permitted$;
\item if $\vec a_1$ is a request or match on $a\in\necessaryset$, then $t$ is \emph{semi-controllable} (in $\mathcal A$) and part of $T^\Box$.
\end{itemize}
Moreover, given $\mathcal A'\subseteq\mathcal A$, if $t$ is semi-controllable and
$\exists\,t' = (\vec{q_2}, \vec {a_2}, \vec{q_2}\!') \in T^\Box_{{\mathcal A}'}$ in $\mathcal A'$ such that $\vec a_2$ is a match,
$\vec q_2, \vec q_2\!' \not\in \textit{Dangling}(\mathcal A')$,
$\ithel{\vec{q}_1{}}{i} = \ithel{\vec{q}_2{}}{i}$, and $\ithel{\vec{a}_1{}}{i} = \ithel{\vec{a}_2{}}{i} = a$, then $t$ is \emph{controllable} in $\mathcal A'$ (via $t'$);
otherwise, $t$ is \emph{uncontrollable} in $\mathcal A'$.
\end{defi}

The algorithm for synthesising an orchestration enforcing agreement of MSCA follows.
The main adaptation of the mpc synthesis of Theorem~\ref{the:mpc} is that transitions are no longer declared uncontrollable, but instead they can be either controllable or semi-controllable.
More importantly, a semi-controllable transition switches from controllable to uncontrollable only after it has been pruned in a previous iteration, in which case its source state becomes bad.
Finally, in this case there are no forbidden states but rather forbidden transitions (i.e.\ requests, according to the property of agreement).

\begin{defi}[MSCA orchestration synthesis, adapted from~\cite{BDGDF17}]\label{def:synthesisorchestration}
Let $\mathcal{A}$ be an MSCA, and let $\mathcal{K}_{0} = \mathcal A$ and $R_{0} = \textit{Dangling}(\mathcal{K}_{0})$.
We let the \emph{orchestration synthesis function} $f_o: \textit{MSCA} \times 2^Q \rightarrow \textit{MSCA} \times 2^Q$ be defined as follows:
\begin{center}
\def\arraystretch{1.2}\begin{tabular}{c@{\hskip 0.5in}r@{\hskip 0.05in}c@{\hskip 0.05in}l}
\multicolumn{4}{l}{
$f_o(\mathcal{K}_{i-1},R_{i-1}) = (\mathcal{K}_{i},R_{i}),
\text{ with }$}\\
& $T_{\mathcal{K}_{i}}$ & = & $T_{\mathcal{K}_{i-1}}\setminus
\{\,(\vec{q} \TRANS{}\vec{q}\,') = t\in T_{\mathcal{K}_{i-1}} \mid (\vec{q}\,'\!\in R_{i-1}\, \vee t \textit{ is a request})\}$
\\
& $R_{i}$ & = & $R_{i-1} \cup
\{\,\vec q \mid (\vec q \TRANS{})\in T^\Box_{\mathcal A} \textit{ is uncontrollable in } \mathcal{K}_{i}\,\} \cup \textit{Dangling}(\mathcal{K}_{i})$
\end{tabular}
\end{center}
\end{defi}

\begin{thm}[MSCA orchestration, adapted from~\cite{BDGDF17}]%
\label{the:def-orchestration}
The orchestration synthesis function $f_o$ is monotone on the cpo $(\textit{MSCA}\times 2^Q, \leq)$ and its \emph{least fixed point} is:
\[
(\mathcal K_{s}, R_{s}) = \sup (\{\,f_o^n(\mathcal K_{0}, R_{0})\mid n \in \mathbb N\,\})
\]

\noindent
The \emph{orchestration} $\mathcal K_{\mathcal{A}}$ of $\mathcal A$ is:
\[
\mathcal K_{\mathcal{A}} =
\left\{\begin{array}{l@{\qquad}l}
    \langle~\rangle & \mbox{if } \vec q_0 \in R_{s} \\
    \langle Q\setminus R_{s}, \vec{q_0}, \permittedset,\necessaryset, A^{o}, T_{\mathcal K_{s}}\!\setminus T', F\setminus R_{s} \rangle & \mbox{otherwise}
\end{array}\right.
\]
where $T' = \{\,t = \vec q\;\TRANSS{}\in\mathcal K_{s} \mid t \text{ is controllable in } \mathcal K_{s},\ \vec q\in R_{s}\,\}$.
\end{thm}

\modiff{We now estimate the complexity of the orchestration synthesis algorithm.
In the synthesis of the mpc, deciding whether a transition is controllable or uncontrollable has a complexity of $\mathcal O(1)$.
On the converse, for the orchestration case, deciding whether a semi-controllable transition is controllable or uncontrollable requires in the worst case to check all transitions of the automaton.
Accordingly, the procedure for computing the set of uncontrollable transitions has an upper-bound complexity of $\mathcal O(|T|^2)$.
Since this is the only difference with respect to the mpc synthesis, a first upper bound of the complexity of the orchestration synthesis is $\mathcal{O}((|T|+|Q|)\times|Q|^2\times|T|^2)$.
The computation of the set of dangling states and uncontrollable transitions could be done in parallel through a single visit of the automaton.
Thus, the upper-bound complexity of the orchestration synthesis can be lowered to be the same as the complexity of the mpc synthesis, i.e.\ $\mathcal{O}((|T|+|Q|)\times|Q|^2)$.
Finally, we want to underline that our complexity estimation refers to the abstract specification of the algorithm, resulting from Definition~\ref{def:synthesisorchestration} and Theorem~\ref{the:def-orchestration}. As already observed for the mpc synthesis, when implementing the algorithm further optimisations could be achieved that can lower its complexity.}


\begin{exa}\label{ex:orchestration}
\modiff{We further continue the running example by discussing the synthesis of the orchestration for the composition
\[
\mathcal A_1 = \textsf{Client}\otimes \textsf{Client} \otimes \textsf{Broker}\otimes \textsf{Hotel} \otimes \textsf{PrivilegedHotel}
\]
The orchestration of $\mathcal A_1$ is depicted in Fig.~\ref{fig:orch} and the time needed to compute it by using FMCAT is reported in Table~\ref{tab:results} (on page~\pageref{tab:results}).
We recall that the orchestration is the largest sub-portion of the composition that is in agreement, i.e.\ in which requests are matched by offers.}

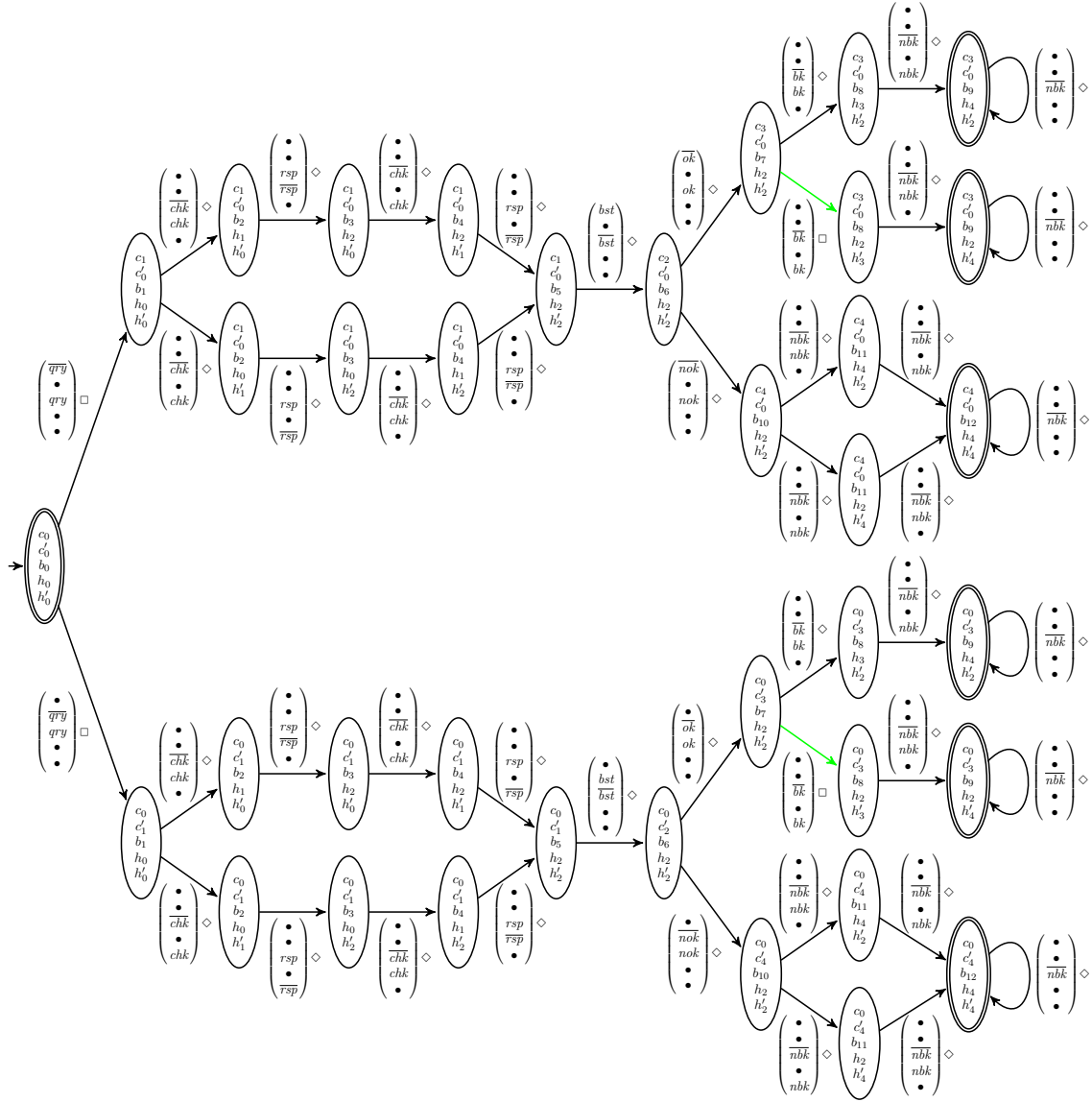
\begin{figure}
\centering
\resizebox{\textwidth}{!}{
\begin{tikzpicture}[->,scale=.5,>=stealth',shorten >=1pt,auto,node distance=1cm,
        semithick, every node/.style={scale=0.5},initial text={},inner sep=0pt, minimum size=0pt]
\tikzstyle{every state}=[fill=white,draw=black,text=black,ellipse]

\node[initial,state,accepting] (0) {$\begin{array}{c}c_{0}\\ c'_{0}\\ b_{0}\\ h_{0}\\ h'_{0}\end{array}$};
\node[state] (1) [above right = 2.85cm and 1cm of 0] {$\begin{array}{c}\,c_{1}\,\\ c'_{0}\\ b_{1}\\ h_{0}\\ h'_{0}\end{array}$};
\node[state] (2) [above right = -.15cm and 1cm of 1] {$\begin{array}{c}\,c_{1}\,\\ c'_{0}\\ b_{2}\\ h_{1}\\ h'_{0}\end{array}$};
\node[state] (3) [right = of 2] {$\begin{array}{c}\,c_{1}\,\\ c'_{0}\\ b_{3}\\ h_{2}\\ h'_{0}\end{array}$};
\node[state] (4) [right = of 3] {$\begin{array}{c}\,c_{1}\,\\ c'_{0}\\ b_{4}\\ h_{2}\\ h'_{1}\end{array}$};
\node[state] (5) [below right = -.15cm and 1cm of 1] {$\begin{array}{c}\,c_{1}\,\\ c'_{0}\\ b_{2}\\ h_{0}\\ h'_{1}\end{array}$};
\node[state] (6) [right = of 5] {$\begin{array}{c}\,c_{1}\,\\ c'_{0}\\ b_{3}\\ h_{0}\\ h'_{2}\end{array}$};
\node[state] (7) [right = of 6] {$\begin{array}{c}\,c_{1}\,\\ c'_{0}\\ b_{4}\\ h_{1}\\ h'_{2}\end{array}$};
\node[state] (8) [below right = -.15cm and 1cm of 4] {$\begin{array}{c}\,c_{1}\,\\ c'_{0}\\ b_{5}\\ h_{2}\\ h'_{2}\end{array}$};
\node[state] (9) [right = of 8] {$\begin{array}{c}c_{2}\\ c'_{0}\\ b_{6}\\ h_{2}\\ h'_{2}\end{array}$};
\node[state] (10) [above right = .75cm and 1cm of 9] {$\begin{array}{c}c_{3}\\ c'_{0}\\ \,b_{7}\,\\ h_{2}\\ h'_{2}\end{array}$};
\node[state] (11) [above right = -.15cm and 1cm of 10] {$\begin{array}{c}c_{3}\\ c'_{0}\\ \,b_{8}\,\\ h_{3}\\ h'_{2}\end{array}$};
\node[state,accepting] (12) [right = of 11] {$\begin{array}{c}c_{3}\\ c'_{0}\\ \,b_{9}\,\\ h_{4}\\ h'_{2}\end{array}$};
\node[state] (13) [below right = -.15cm and 1cm of 10] {$\begin{array}{c}c_{3}\\ c'_{0}\\ \,b_{8}\,\\ h_{2}\\ h'_{3}\end{array}$};
\node[state,accepting] (14) [right = of 13] {$\begin{array}{c}c_{3}\\ c'_{0}\\ \,b_{9}\,\\ h_{2}\\ h'_{4}\end{array}$};
\node[state] (15) [below right = .75cm and 1cm of 9] {$\begin{array}{c}\,c_{4}\,\\ c'_{0}\\ b_{10}\\ h_{2}\\ h'_{2}\end{array}$};
\node[state] (16) [above right = -.15cm and 1cm of 15] {$\begin{array}{c}\,c_{4}\,\\ c'_{0}\\ b_{11}\\ h_{4}\\ h'_{2}\end{array}$};
\node[state] (17) [below right = -.15cm and 1cm of 15] {$\begin{array}{c}c_{4}\\ c'_{0}\\ b_{11}\\ h_{2}\\ h'_{4}\end{array}$};
\node[state,accepting] (18) [below right = -.15cm and 1.15cm of 16] {$\begin{array}{c}c_{4}\\ c'_{0}\\ b_{12}\\ h_{4}\\ h'_{4}\end{array}$};
\node[state] (19) [below right = 2.85cm and 1cm of 0] {$\begin{array}{c}c_{0}\\ \,c'_{1}\,\\ b_{1}\\ h_{0}\\ h'_{0}\end{array}$};
\node[state] (20) [above right = -.15cm and 1cm of 19] {$\begin{array}{c}c_{0}\\ \,c'_{1}\,\\ b_{2}\\ h_{1}\\ h'_{0}\end{array}$};
\node[state] (21) [right = of 20] {$\begin{array}{c}c_{0}\\ \,c'_{1}\,\\ b_{3}\\ h_{2}\\ h'_{0}\end{array}$};
\node[state] (22) [right = of 21] {$\begin{array}{c}c_{0}\\ \,c'_{1}\,\\ b_{4}\\ h_{2}\\ h'_{1}\end{array}$};
\node[state] (23) [below right = -.15cm and 1cm of 19] {$\begin{array}{c}c_{0}\\ \,c'_{1}\,\\ b_{2}\\ h_{0}\\ h'_{1}\end{array}$};
\node[state] (24) [right = of 23] {$\begin{array}{c}c_{0}\\ \,c'_{1}\,\\ b_{3}\\ h_{0}\\ h'_{2}\end{array}$};
\node[state] (25) [right = of 24] {$\begin{array}{c}c_{0}\\ \,c'_{1}\,\\ b_{4}\\ h_{1}\\ h'_{2}\end{array}$};
\node[state] (26) [below right = -.15cm and 1cm of 22] {$\begin{array}{c}c_{0}\\ \,c'_{1}\,\\ b_{5}\\ h_{2}\\ h'_{2}\end{array}$};
\node[state] (27) [right = of 26] {$\begin{array}{c}c_{0}\\ c'_{2}\\ b_{6}\\ h_{2}\\ h'_{2}\end{array}$};
\node[state] (28) [above right = .75cm and 1cm of 27] {$\begin{array}{c}c_{0}\\ c'_{3}\\ \,b_{7}\,\\ h_{2}\\ h'_{2}\end{array}$};
\node[state] (29) [above right = -.15cm and 1cm of 28] {$\begin{array}{c}c_{0}\\ c'_{3}\\ \,b_{8}\,\\ h_{3}\\ h'_{2}\end{array}$};
\node[state,accepting] (30) [right = of 29] {$\begin{array}{c}c_{0}\\ c'_{3}\\ \,b_{9}\,\\ h_{4}\\ h'_{2}\end{array}$};
\node[state] (31) [below right = -.15cm and 1cm of 28] {$\begin{array}{c}c_{0}\\ c'_{3}\\ \,b_{8}\,\\ h_{2}\\ h'_{3}\end{array}$};
\node[state,accepting] (32) [right = of 31] {$\begin{array}{c}c_{0}\\ c'_{3}\\ \,b_{9}\,\\ h_{2}\\ h'_{4}\end{array}$};
\node[state] (33) [below right = .75cm and 1cm of 27] {$\begin{array}{c}c_{0}\\ c'_{4}\\ b_{10}\\ h_{2}\\ h'_{2}\end{array}$};
\node[state] (34) [above right = -.15cm and 1cm of 33] {$\begin{array}{c}c_{0}\\ c'_{4}\\ b_{11}\\ h_{4}\\ h'_{2}\end{array}$};
\node[state] (35) [below right = -.15cm and 1cm of 33] {$\begin{array}{c}c_{0}\\ c'_{4}\\ b_{11}\\ h_{2}\\ h'_{4}\end{array}$};
\node[state,accepting] (36) [below right = -.15cm and 1.15cm of 34] {$\begin{array}{c}c_{0}\\ c'_{4}\\ b_{12}\\ h_{4}\\ h'_{4}\end{array}$};

\path
(0) edge[] node[above,xshift=-25pt,yshift=-10pt]{$\left(\!\!\begin{array}{c}\overline{\textit{qry}}\\ \blk\\ \textit{qry}\\ \blk\\ \blk\end{array}\!\!\right)\!\Necessary$} (1)
(1) edge[] node[above,xshift=-5pt,yshift=5pt]{$\left(\!\!\begin{array}{c}\blk\\ \blk\\ \overline{\textit{chk}}\\ \textit{chk}\\ \blk\end{array}\!\!\right)\!\Permitted$} (2)
(2) edge[] node[above,yshift=5pt]{$\left(\!\!\begin{array}{c}\blk\\ \blk\\ \textit{rsp}\\ \overline{\textit{rsp}}\\ \blk\end{array}\!\!\right)\!\Permitted$} (3)
(3) edge[] node[above,yshift=5pt]{$\left(\!\!\begin{array}{c}\blk\\ \blk\\ \overline{\textit{chk}}\\ \blk\\ \textit{chk}\end{array}\!\!\right)\!\Permitted$} (4)
(4) edge[] node[above,xshift=10pt,yshift=5pt]{$\left(\!\!\begin{array}{c}\blk\\ \blk\\ \textit{rsp}\\ \blk\\ \overline{\textit{rsp}}\end{array}\!\!\right)\!\Permitted$} (8)
(1) edge[] node[below,xshift=-5pt,yshift=-5pt]{$\left(\!\!\begin{array}{c}\blk\\ \blk\\ \overline{\textit{chk}}\\ \blk\\ \textit{chk}\end{array}\!\!\right)\!\Permitted$} (5)
(5) edge[] node[below,yshift=-5pt]{$\left(\!\!\begin{array}{c}\blk\\ \blk\\ \textit{rsp}\\ \blk\\ \overline{\textit{rsp}}\end{array}\!\!\right)\!\Permitted$} (6)
(6) edge[] node[below,yshift=-5pt]{$\left(\!\!\begin{array}{c}\blk\\ \blk\\ \overline{\textit{chk}}\\ \textit{chk}\\ \blk\end{array}\!\!\right)\!\Permitted$} (7)
(7) edge[] node[below,xshift=10pt,yshift=-5pt]{$\left(\!\!\begin{array}{c}\blk\\ \blk\\ \textit{rsp}\\ \overline{\textit{rsp}}\\ \blk\end{array}\!\!\right)\!\Permitted$} (8)
(8) edge[] node[above,yshift=5pt]{$\left(\!\!\begin{array}{c}\textit{bst}\\ \blk\\ \overline{\textit{bst}}\\ \blk\\ \blk\end{array}\!\!\right)\!\Permitted$} (9)
(9) edge[] node[above,xshift=-15pt,yshift=-5pt]{$\left(\!\!\begin{array}{c}\overline{\textit{ok}}\\ \blk\\ \textit{ok}\\ \blk\\ \blk\end{array}\!\!\right)\!\Permitted$} (10)
(10) edge[] node[above,xshift=-5pt,yshift=5pt]{$\left(\!\!\begin{array}{c}\blk\\ \blk\\ \overline{\textit{bk}}\\ \textit{bk}\\ \blk\end{array}\!\!\right)\!\Permitted$} (11)
(11) edge[] node[above,yshift=5pt]{$\left(\!\!\begin{array}{c}\blk\\ \blk\\ \overline{\textit{nbk}}\\ \blk\\ \textit{nbk}\end{array}\!\!\right)\!\Permitted$} (12)
(12) edge[loop,looseness=5,in=315,out=45] node[right,xshift=3pt]{$\left(\!\!\begin{array}{c}\blk\\ \blk\\ \overline{\textit{nbk}}\\ \blk\\ \blk\end{array}\!\!\right)\!\Permitted$} (12)
(10) edge[draw=green] node[below,xshift=-5pt,yshift=-5pt]{$\left(\!\!\begin{array}{c}\blk\\ \blk\\ \overline{\textit{bk}}\\ \blk\\ \textit{bk}\end{array}\!\!\right)\!\Necessary$} (13)
(13) edge[] node[above,yshift=5pt]{$\left(\!\!\begin{array}{c}\blk\\ \blk\\ \overline{\textit{nbk}}\\ \textit{nbk}\\ \blk\end{array}\!\!\right)\!\Permitted$} (14)
(14) edge[loop,looseness=5,in=315,out=45] node[right,xshift=3pt]{$\left(\!\!\begin{array}{c}\blk\\ \blk\\ \overline{\textit{nbk}}\\ \blk\\ \blk\end{array}\!\!\right)\!\Permitted$} (14)
(9) edge[] node[below,xshift=-15pt,yshift=-5pt]{$\left(\!\!\begin{array}{c}\overline{\textit{nok}}\\ \blk\\ \textit{nok}\\ \blk\\ \blk\end{array}\!\!\right)\!\Permitted$} (15)
(15) edge[] node[above,xshift=-5pt,yshift=5pt]{$\left(\!\!\begin{array}{c}\blk\\ \blk\\ \overline{\textit{nbk}}\\ \textit{nbk}\\ \blk\end{array}\!\!\right)\!\Permitted$} (16)
(15) edge[] node[below,xshift=-5pt,yshift=-5pt]{$\left(\!\!\begin{array}{c}\blk\\ \blk\\ \overline{\textit{nbk}}\\ \blk\\ \textit{nbk}\end{array}\!\!\right)\!\Permitted$} (17)
(16) edge[] node[above,xshift=10pt,yshift=5pt]{$\left(\!\!\begin{array}{c}\blk\\ \blk\\ \overline{\textit{nbk}}\\ \blk\\ \textit{nbk}\end{array}\!\!\right)\!\Permitted$} (18)
(17) edge[] node[below,xshift=10pt,yshift=-5pt]{$\left(\!\!\begin{array}{c}\blk\\ \blk\\ \overline{\textit{nbk}}\\ \textit{nbk}\\ \blk\end{array}\!\!\right)\!\Permitted$} (18)
(18) edge[loop,looseness=5,in=315,out=45] node[right,xshift=3pt]{$\left(\!\!\begin{array}{c}\blk\\ \blk\\ \overline{\textit{nbk}}\\ \blk\\ \blk\end{array}\!\!\right)\!\Permitted$} (18)
(0) edge[] node[below,xshift=-25pt,yshift=10pt]{$\left(\!\!\begin{array}{c}\blk\\ \overline{\textit{qry}}\\ \textit{qry}\\ \blk\\ \blk\end{array}\!\!\right)\!\Necessary$} (19)
(19) edge[] node[above,xshift=-5pt,yshift=5pt]{$\left(\!\!\begin{array}{c}\blk\\ \blk\\ \overline{\textit{chk}}\\ \textit{chk}\\ \blk\end{array}\!\!\right)\!\Permitted$} (20)
(20) edge[] node[above,yshift=5pt]{$\left(\!\!\begin{array}{c}\blk\\ \blk\\ \textit{rsp}\\ \overline{\textit{rsp}}\\ \blk\end{array}\!\!\right)\!\Permitted$} (21)
(21) edge[] node[above,yshift=5pt]{$\left(\!\!\begin{array}{c}\blk\\ \blk\\ \overline{\textit{chk}}\\ \blk\\ \textit{chk}\end{array}\!\!\right)\!\Permitted$} (22)
(22) edge[] node[above,xshift=10pt,yshift=5pt]{$\left(\!\!\begin{array}{c}\blk\\ \blk\\ \textit{rsp}\\ \blk\\ \overline{\textit{rsp}}\end{array}\!\!\right)\!\Permitted$} (26)
(19) edge[] node[below,xshift=-5pt,yshift=-5pt]{$\left(\!\!\begin{array}{c}\blk\\ \blk\\ \overline{\textit{chk}}\\ \blk\\ \textit{chk}\end{array}\!\!\right)\!\Permitted$} (23)
(23) edge[] node[below,yshift=-5pt]{$\left(\!\!\begin{array}{c}\blk\\ \blk\\ \textit{rsp}\\ \blk\\ \overline{\textit{rsp}}\end{array}\!\!\right)\!\Permitted$} (24)
(24) edge[] node[below,yshift=-5pt]{$\left(\!\!\begin{array}{c}\blk\\ \blk\\ \overline{\textit{chk}}\\ \textit{chk}\\ \blk\end{array}\!\!\right)\!\Permitted$} (25)
(25) edge[] node[below,xshift=10pt,yshift=-5pt]{$\left(\!\!\begin{array}{c}\blk\\ \blk\\ \textit{rsp}\\ \overline{\textit{rsp}}\\ \blk\end{array}\!\!\right)\!\Permitted$} (26)
(26) edge[] node[above,yshift=5pt]{$\left(\!\!\begin{array}{c}\blk\\ \textit{bst}\\ \overline{\textit{bst}}\\ \blk\\ \blk\end{array}\!\!\right)\!\Permitted$} (27)
(27) edge[] node[above,xshift=-15pt,yshift=-5pt]{$\left(\!\!\begin{array}{c}\blk\\ \overline{\textit{ok}}\\ \textit{ok}\\ \blk\\ \blk\end{array}\!\!\right)\!\Permitted$} (28)
(28) edge[] node[above,xshift=-5pt,yshift=5pt]{$\left(\!\!\begin{array}{c}\blk\\ \blk\\ \overline{\textit{bk}}\\ \textit{bk}\\ \blk\end{array}\!\!\right)\!\Permitted$} (29)
(29) edge[] node[above,yshift=5pt]{$\left(\!\!\begin{array}{c}\blk\\ \blk\\ \overline{\textit{nbk}}\\ \blk\\ \textit{nbk}\end{array}\!\!\right)\!\Permitted$} (30)
(30) edge[loop,looseness=5,in=315,out=45] node[right,xshift=3pt]{$\left(\!\!\begin{array}{c}\blk\\ \blk\\ \overline{\textit{nbk}}\\ \blk\\ \blk\end{array}\!\!\right)\!\Permitted$} (30)
(28) edge[draw=green] node[below,xshift=-5pt,yshift=-5pt]{$\left(\!\!\begin{array}{c}\blk\\ \blk\\ \overline{\textit{bk}}\\ \blk\\ \textit{bk}\end{array}\!\!\right)\!\Necessary$} (31)
(31) edge[] node[above,yshift=5pt]{$\left(\!\!\begin{array}{c}\blk\\ \blk\\ \overline{\textit{nbk}}\\ \textit{nbk}\\ \blk\end{array}\!\!\right)\!\Permitted$} (32)
(32) edge[loop,looseness=5,in=315,out=45] node[right,xshift=3pt]{$\left(\!\!\begin{array}{c}\blk\\ \blk\\ \overline{\textit{nbk}}\\ \blk\\ \blk\end{array}\!\!\right)\!\Permitted$} (32)
(27) edge[] node[below,xshift=-15pt,yshift=-5pt]{$\left(\!\!\begin{array}{c}\blk\\ \overline{\textit{nok}}\\ \textit{nok}\\ \blk\\ \blk\end{array}\!\!\right)\!\Permitted$} (33)
(33) edge[] node[above,xshift=-5pt,yshift=5pt]{$\left(\!\!\begin{array}{c}\blk\\ \blk\\ \overline{\textit{nbk}}\\ \textit{nbk}\\ \blk\end{array}\!\!\right)\!\Permitted$} (34)
(33) edge[] node[below,xshift=-5pt,yshift=-5pt]{$\left(\!\!\begin{array}{c}\blk\\ \blk\\ \overline{\textit{nbk}}\\ \blk\\ \textit{nbk}\end{array}\!\!\right)\!\Permitted$} (35)
(34) edge[] node[above,xshift=10pt,yshift=5pt]{$\left(\!\!\begin{array}{c}\blk\\ \blk\\ \overline{\textit{nbk}}\\ \blk\\ \textit{nbk}\end{array}\!\!\right)\!\Permitted$} (36)
(35) edge[] node[below,xshift=10pt,yshift=-5pt]{$\left(\!\!\begin{array}{c}\blk\\ \blk\\ \overline{\textit{nbk}}\\ \textit{nbk}\\ \blk\end{array}\!\!\right)\!\Permitted$} (36)
(36) edge[loop,looseness=5,in=315,out=45] node[right,xshift=3pt]{$\left(\!\!\begin{array}{c}\blk\\ \blk\\ \overline{\textit{nbk}}\\ \blk\\ \blk\end{array}\!\!\right)\!\Permitted$} (36);
\end{tikzpicture}
}
\caption{\label{fig:orch}\color{black}{Orchestration of 
$\textsf{Client}\otimes \textsf{Client} \otimes \textsf{Broker}\otimes \textsf{Hotel} \otimes \textsf{PrivilegedHotel}$}}
\end{figure}

\modiff{From the initial (and final) state there are two possible evolutions: either one of the clients is served while the other one does not interact.
Without loss of generality, assume that the first client is served.
The orchestration continues with the broker enquiring the hotels (in both possible orders).
After these enquiries, the reached state is $\vec q = (c_1,c'_0,b_5,h_2,h'_2)$.
From $\vec q$, the broker sends the best offer received from one of the hotels to the client, and the client decides whether or not to accept this best offer.
The broker then communicates the selected choice to the hotels it interacted with.
Note that in the orchestration it is possible that the client does not book any room.}

\modiff{We now explain why the mpc is empty (cf.\ Example~\ref{ex:mpc}).
First, note that $\vec q$ must be traversed to reach a final successful state.
The composition $\mathcal A_1$ (which is not displayed for space limitations) contains the transition $t = \vec q\,\TRANSS{(\blk,\blk,\blk,\blk, \textit{bk})\Box}\,(c_1,c'_0,b_5,h_2,h'_3)$.
The state $\vec q$ is then forbidden, since its outgoing transition $t$ is uncontrollable and cannot be pruned. It follows that  all states traversed from the initial state to $\vec q$ would eventually become dangling during the mpc synthesis, and thus the mpc is empty.}

\modiff{On the converse, for the case of synthesising the orchestration, we have that $t$ is \emph{semi-controllable} and it is \emph{controllable} via $t' = (c_3,c'_0,b_7,h_2,h'_2)\,\TRANSS{(\blk,\blk,\overline{\textit{bk}},\blk, \textit{bk})\Box}\,(c_1,c'_0,b_8,h_2,h'_3)$.
Thus, $t$ is pruned by the orchestration synthesis algorithm.
Intuitively, in the orchestration there exists a trace in which the necessary request $\textit{bk}$ is matched.}
\end{exa}


\modiff{This example shows that the notion of semi-controllability is best suited for necessary requests of service contracts.
We argue that semi-controllability is not specific to the context of service contracts; rather it is independent of the used formalism and can be applied in other contexts as well.
Semi-controllability can be interpreted as the \lq existentially quantified\rq\ counterpart of the universally quantified notion of uncontrollability, originally stemming from Supervisory Control Theory, in much the same way that Computation Tree Logic allows existential quantification of paths that can only be universally quantified in Linear Temporal Logic.}

\modiff{However, in the next section we will see that the notion of semi-controllability is too relaxed for the case of choreography, which thus demands a revisited version.}

\subsection{On encoding semi-controllability.}
We now show, by means of an example adapted from~\cite{BBL19}, that the encoding of an automaton $\mathcal A$ with semi-controllable actions into an automaton $\mathcal A'$ without, such that the same synthesised orchestrations are obtained, results in an exponential blow-up of the state space. More precisely, the encoding is intended to preserve safety: the orchestration of $\mathcal A$ equals that of $\mathcal A'$.


\modiff{
The encoding is sketched in Fig.~\ref{fig:Aprime}. Intuitively, the encoded automaton $\mathcal A'$ is obtained by first applying the following construction to the automaton $\mathcal A$ from Fig.~\ref{fig:comp}~(right):
\begin{quote}
if the synchronisation on a specific semi-controllable action~$a$ occurs in $n$~different transitions in $\mathcal A$ (two in our example), then the encoding creates an automaton $\mathcal A'$ that is the union of $2^n - 1$ automata (three in our example), which are obtained by all possible combinations of pruning a subset of the $n$~semi-controllable transitions of~$\mathcal A$, minus the one in which all $n$~semi-controllable transitions are pruned;
\end{quote}
and then turning all semi-controllable transitions into uncontrollable transitions.
}

\begin{figure}
\centering\begin{tabular}{c}
\def\objectstyle{\scriptstyle}\def\labelstyle{\scriptstyle}
\xymatrix@C=.55cm@R=1cm{
*{} & *{} & *{} & *{} & *{} & *+[F] {(\cdot,\cdot)} \ar@{.>}@(l,u)[dlllll] \ar@{.>}[d] \ar@{.>}@(r,u)[drrrrr] \\ 
*+[F] {(\cdot,\cdot)} \ar@{.>}[rr]^{(\blk,\overline{b})\Permitted} & *{} & *+[F] {(\cdot,\cdot)} \ar@[red][drr]_*+<-1.25em>[@!-37]{~(a,\overline{a})\Boxurgent} & *{} & *{} & *+[F] {(\cdot,\cdot)} \ar@[red][drr]_*+<-1.25em>[@!-37]{~(a,\overline{a})\Boxurgent} \ar@{.>}[rr]^{(\blk,\overline{b})\Permitted} & *{} & *+[F] {(\cdot,\cdot)} \ar@[red][drr]_*+<-1.25em>[@!-37]{~(a,\overline{a})\Boxurgent} & *{} & *{} & *+[F] {(\cdot,\cdot)} \ar@[red][drr]_*+<-1.25em>[@!-37]{~(a,\overline{a})\Boxurgent} \ar@{.>}[rr]^{(\blk,\overline{b})\Permitted} & *{} & *+[F] {(\cdot,\cdot)} \\
*{} & *{} & *{} & *{} & *+[F] {(\cdot,\cdot)} & *{} & *{} & *+[F] {(\cdot,\cdot)} & *{} & *+[F] {(\cdot,\cdot)} & *{} & *{} & *+[F] {(\cdot,\cdot)} \\
}
\end{tabular}
\caption{\label{fig:Aprime}Automaton ${\mathcal{A}'}$ uses uncontrollable transitions to encode automaton ${\mathcal{A}}$ from Fig.~\ref{fig:comp}}
\end{figure}
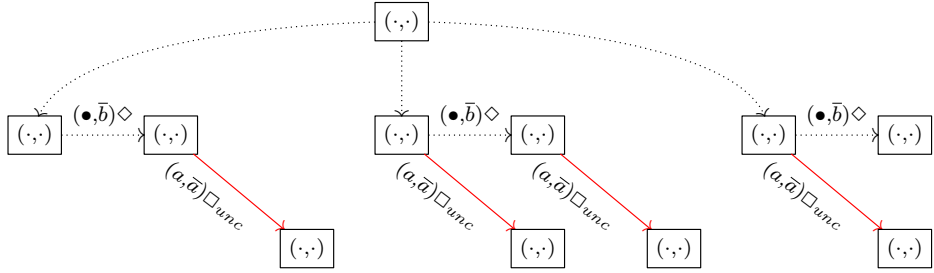

We now explain why, without knowing a priori the set of forbidden and successful states, it is impossible to provide a more efficient encoding and refer to~\cite[Theorems~3 and~4]{BBL19} for a formal account. Assume, by contradiction, that there exists an encoding that results in a \lq smaller\rq\ automaton $\mathcal A''$, in which one of the $2^n - 1$ combinations of pruned transitions (say, $P$) is discarded.
It then suffices to specify as a counterexample a property in $\mathcal A$ such that all source states of transitions in $P$ are forbidden and all target states of the remaining semi-controllable transitions are successful.
The synthesis of $\mathcal A$ against such a property would prune exactly the semi-controllable transitions in $P$.
However, in the synthesis of $\mathcal A''$ such an orchestration would not be present, a contradiction.

\section{Synthesis of Choreographies}\label{sect:choreographysynthesis}

In the previous section, we have seen that the orchestration of MSCA  is similar to a most permissive controller.
The orchestrator is however implicit, in the sense that its interactions with the principals are hidden.
Basically, one could assume that before interacting, each principal expects a message from the orchestrator and answers with an acknowledgement after the interaction terminates.
The main intuition behind switching from an orchestrated to a choreographic coordination of contracts is that there is no longer the need for such \lq hidden\rq\ interactions.
Ideally, the principals moving autonomously are able to accomplish the behaviour foreseen by the synthesis, which in this case acts as a global type.
Differently from the traditional choreographic approach, where the starting point is a global type, in MSCA the global type is synthesised  automatically.

The requirements for ensuring that the synthesised automaton is a (form of) choreography were
 studied in~\cite{BDFT15,Lange:2015:CMG:2676726.2676964}.
Roughly, they amount to the so-called \emph{branching condition} requiring that principals perform their offers/outputs independently of the other principals in the composition.
To formalise this notion, we let $\textit{snd}(\vec a) = i$ when $\vec a$ is a match action or an offer action and $\ithel{\vec{a}}{ i} \in \Oset$.

\begin{defi}[Branching condition~\cite{BDFT15}]\label{def:branchingcondition}
An MSCA $\mathcal A$ satisfies the \emph{branching condition} if and only if the following holds for each pair of states $\vec q_1$, $\vec q_2$ reachable in $\mathcal A$:
\[
\forall \text{$\vec{a}$ match action } .\ (\vec{q}_1 \TRANSS{\vec{a}} \wedge \textit{snd}(\vec{a})=  i \wedge \ithel{\vec{q_1}}{i}=\ithel{\vec{q_2}}{i}) \text{ implies }  \vec{q}_2 \TRANSS{\vec{a}} .
\]
\end{defi}

The branching condition is related to a phenomenon known as \lq state sharing\rq\ in other coordination models (cf., e.g.,~\cite{BCHK17}) according to which system components can influence potential synchronisations through their local (component) states
even if they are not involved in the actual global (system) transition.

In~\cite{BDFT15}, it is proved that the synthesised automaton  corresponds to a well-behaving choreography if and only if it satisfies the branching condition and is strongly safe. Notably, in case the two conditions are not satisfied, that paper does not provide any algorithm for automatically synthesising a choreography; rather, the contracts have to be manually amended. Instead, in the remainder of this section, we introduce a novel algorithm for automatically synthesising a well-behaving choreography.
Note that, differently from the orchestration and the controller synthesis, in this case there could be more than one possible choreography (cf.\ Example~\ref{ex:choreography}).

The property to be enforced during the synthesis is strong agreement: all offers and requests have to be matched, because  all messages have to be read (i.e.\ offers matched).
Moreover, in the case of choreography, service contract requests are always permitted whereas service contract offers can be necessary.
That is, the roles of service requests and offers are swapped with respect to the case of orchestration.

In principle, the synthesis could trivially introduce a coordinator component and its interactions to coordinate the principals. However, this would reduce the choreography to a centralised coordination of contracts.
To prevent this, the synthesis can only remove and never add behaviour.
Hence, a choreography can only be synthesised if 
all principals are capable of interacting on their own without resorting to a central coordinator.

Similarly to orchestration synthesis, indicating transitions as either controllable or uncontrollable does not suffice for synthesising a choreography. Moreover, the notion of semi-controllability introduced for the orchestration case does not suffice for expressing necessary offers.
Indeed, orchestration synthesis does not ensure the branching condition to be satisfied by the synthesised automaton, as the following example shows.

\begin{exa}\label{ex:semicontrollabilityandchoregraphy}
In Fig.~\ref{fig:diff}, a fragment of a service composition is shown. Two global states are depicted, and in both the first service, say $\Alice$, is in its initial local state (say, $q_0$).
$\Alice$ performs an output (i.e.\ offer) $\overline a$ that can be directed to either $\Bob$ (second service) or $\Carol$ (third service), from the initial global state, or only to $\Bob$ from the other state.
It is possible to reach either a successful ({\color{green!60!black}$\cmark$}) or a bad ({\color{red}$\xmark$}) state, left unspecified for the moment.
Notably, the output of $\Alice$ is neither controllable, nor uncontrollable, nor semi-controllable by the synthesis.

\begin{figure}[t]
\centering\begin{tabular}{c}
\def\objectstyle{\scriptstyle}\def\labelstyle{\scriptstyle}
\xymatrix@C=1.2cm@R=1cm{
*{} & *+[F] {\text{\txt{$(q_0,\cdot,\cdot)$}}} \ar[dl]!U^*+<-.5em>[@!29]{(\overline a,a,\blk)} \ar[dr]!U_*+<-.5em>[@!-29]{(\overline a,\blk,a)} \ar[r]^{\ldots} & *+[F] {\text{\txt{$(q_0,\cdot,\cdot)$}}} \ar[dr]!U_*+<-.5em>[@!-29]{(\overline a,a,\blk)} & {}\\ 
{\phantom{\color{green!60!black}\text{\normalsize$\cmark$}\color{red}\text{\normalsize$\xmark$}}} & *{} & {\color{red}\text{\normalsize$\phantom{-}\xmark$}} & {\color{green!60!black}\text{\normalsize$\phantom{-}\cmark$}}
}
\end{tabular}
\caption{\label{fig:diff}Fragment of a possible service composition}
\end{figure}
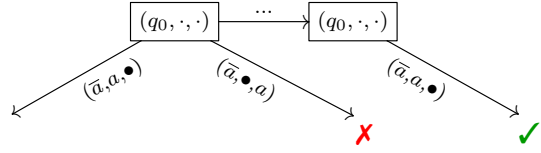

Now assume that the $\overline a$ is controllable and from the initial global state both interactions eventually lead to a bad state ({\color{red}$\xmark$}).
In this case, those transitions are pruned by the synthesis, and the resulting automaton is erroneously approved.
Indeed, $\Alice$ has no mean to understand when her output $\overline a$ is enabled, because she has not changed state.
The branching condition, which is necessary for obtaining a well-behaving choreography, would be violated.
Note that this would happen also if $\overline a$ were semi-controllable.
In fact, to satisfy the branching condition, the synthesis should remove all outputs $\overline a$.

Conversely, assume that the $\overline a$ is uncontrollable and that it is possible from the initial global state to reach a successful state ({\color{green!60!black}$\cmark$}) if the message $\overline a$ is received by $\Bob$.
In this case, it would not be possible to prune the transition from the initial state leading to {\color{red}$\xmark$}, because it is also uncontrollable.
The synthesis would thus be empty, an erroneous rejection, because a choreography exists in which $\Alice$ autonomously interacts with $\Bob$.
\end{exa}

In conclusion, a necessary action is rendered neither as uncontrollable nor as semi-controllable, and permitted actions require extra pruning operations during the synthesis.
A novel notion of semi-controllability for a necessary action is required, which is weaker than uncontrollable but stronger than the semi-controllable notion used in the synthesis of orchestration.

Basically, for the choreography synthesis, a (semi-controllable) necessary transition $t = (\vec{q} \TRANS{\vec a_1}) \in T^\Box$ is detected to be uncontrollable
if and only if no necessary transition $t'=(\vec{q} \TRANS{\vec a_2}) \in T^\Box$ exists from the same source state such that in both $t$ and $t'$ the same offer is provided by the same principal, but possibly with different receivers. We now define this formally.

\begin{defi}\label{def:choreographycontrollability}
Let $\mathcal A$ be an MSCA and let $t=(\vec q, \vec a_1, \vec q_1\!') \in T_{\mathcal A}$. Then:
\begin{itemize}
\item[-] if $\vec a_1$ is an action on $a\in\permittedset$, then $t$ is \emph{controllable} (in $\mathcal A$);
\item[-] if $\vec a_1$ is an offer or match on $a\in\necessaryset$, then $t$ is \emph{semi-controllable} (in $\mathcal A$).
\end{itemize}
Moreover, given $\mathcal A'\subseteq\mathcal A$, if $t$ is semi-controllable and
$\exists\,t' = (\vec q, \vec a_2, \vec q_2\!')\in T^\Box_{{\mathcal A}'}$ 
such that $\vec a_2$ is a match, $\vec q, \vec q_2\!' \not\in \textit{Dangling}(\mathcal A')$, and $\ithel{\vec{a}_1{}}{i} = \ithel{\vec{a}_2{}}{i}$ where $i = snd(\vec a_1)$, then $t$ is \emph{controllable} in $\mathcal A'$ (via $t'$);
otherwise, $t$ is \emph{uncontrollable} in $\mathcal A'$.
\end{defi}

Hence, again a necessary transition is a particular type of transition that switches from being controllable to uncontrollable in case a condition on the global automaton is not met. Note that this condition is stronger than the one required for the case of orchestration (semi-controllability), because for the case of choreography transitions $t$ and $t'$ in Definition~\ref{def:choreographycontrollability} share the source state. Moreover, also in this case it can be shown that the encoding of this type of semi-controllable transition into an uncontrollable one would result in an exponential growth of the state space of the model.

Similarly to the orchestration synthesis in Definition~\ref{def:synthesisorchestration}, when a semi-controllable transition previously removed by the synthesis switches from controllable to uncontrollable, its source state is detected to be bad.
Apart from the different notion of semi-controllability, another difference with respect to the orchestration synthesis is that the transitions violating the branching condition must also be removed.
\modiff{Depending on which transitions violating the branching condition are pruned at a certain iteration, different choreographies can be obtained (cf.\ Example~\ref{ex:choreography}).
Indeed, a maximal choreography is not always guaranteed to exist (as is the case for the running example).
A concrete implementation should fix the criterion under which transitions are selected for the set $\hat{T}_{\mathcal{K}_{i},R_{i}}$ (cf.\ Definition~\ref{def:synthesischoreography}).}

Finally, according to the property of strong agreement, both request and offer transitions are forbidden.
The formalisation is provided next.

\begin{defi}[MSCA choreography synthesis]\label{def:synthesischoreography}
Let $\mathcal{A}$ be an MSCA, and let $\mathcal{K}_{0} = \mathcal{A}$ and $R_{0} = \textit{Dangling}(\mathcal{K}_{0})$.
We let a \emph{choreography synthesis function} $f_c: \textit{MSCA} \times 2^Q \rightarrow \textit{MSCA} \times 2^Q$ be defined as follows:
\vspace{0.5em}
\begin{center}
\def\arraystretch{1.2}
\begin{tabular}{c@{\hskip 0.15in}r@{\hskip 0.05in}c@{\hskip 0.05in}l}
\multicolumn{4}{l}{
$f_c(\mathcal{K}_{i-1},R_{i-1})=(\mathcal{K}_{i},R_{i}), 
\text{ with }$}\\
& $T_{\mathcal{K}_{i}}$ & \!=\! & $T_{\mathcal{K}_{i-1}}\!\setminus
(\{\,(\vec{q} \TRANS{} \vec q \, ' ) = t \in T_{\mathcal K_{i-1}} \mid \vec q \, '\in R_{i-1} \vee t \textit{ is a request or offer}\,\}  \cup \hat{T}_{\mathcal{K}_{i-1},R_{i-1}})$\\
& $R_{i}$ & \!=\! & $R_{i-1} \cup \{\,\vec q \mid (\vec q \TRANS{})\in T_{\mathcal A} \textit{ is uncontrollable in } \mathcal{K}_{i} \,\} \cup \textit{Dangling}(\mathcal{K}_{i})$\\
\end{tabular}
\end{center}
%
%
%
%
%
%
\modiff{where, at each iteration $i$,
\[\hat{T}_{\mathcal{K}_{i},R_{i}} \subseteq T_{bc} = \{\,(\vec q_1 \TRANS{\vec a}) \in T_{\mathcal K_{i}} \mid \exists\,\vec{q_2} : (snd(\vec{a}) = j \wedge \ithel{\vec{q_1}}{j} = \ithel{\vec{q_2}}{j}) \wedge (\vec q_2 \TRANS{\vec a}) \not\in T_{\mathcal K_{i}} \wedge \vec{q_1},\vec{q_2} \not\in R_{i}\,\}\]
and whenever
$f_c(\mathcal{K}_{i},R_{i}) = (\mathcal{K}_{i},R_{i})$
then
$T_{bc} = \emptyset$.}
%
%
\end{defi}

\begin{restatable}[MSCA choreography]{thm}{thesynthesischoreography}\label{the:synthesischoreography}
A choreography synthesis function $f_c$ is monotone on the cpo $(\textit{MSCA} \times 2^Q, \leq)$ and its \emph{least fixed point} is:
\[
(\mathcal K_{s}, R_{s}) = \sup (\{\,f_c^n(\mathcal K_{0}, R_{0})\mid n \in \mathbb N\,\})
\]

\noindent
A \emph{choreography} $\mathcal K_{\mathcal{A}}$ of $\mathcal A$ is:
\[\mathcal K_{\mathcal{A}} =
\left\{\begin{array}{l@{\qquad}l}
          \langle~\rangle & \mbox{if } \vec q_0 \in R_{s} \\
          \langle Q\setminus R_{s}, \vec{q_0}, \permittedset,\necessaryset, A^{o}, T_{\mathcal K_{s}}\!\setminus T', F\setminus R_{s} \rangle & \mbox{otherwise}
\end{array}\right.\]
where $T' = \{\,t = \vec q\; \TRANSS{}\in\mathcal K_{s} \mid t \text{ is controllable in } \mathcal K_{s},\ \vec q\in R_{s}\,\}$.

\medskip
Moreover, $\mathcal K_{\mathcal{A}}$ satisfies the branching condition.
\end{restatable}
\begin{proof}
The algorithm terminates because at each iteration either some transition is pruned or a state becomes forbidden, and both sets of transitions and states are finite.
We now prove that the synthesised automaton is (i)~non-blocking, (ii)~controllable, 
(iii)~strongly safe, and (iv)~satisfies the branching condition.
In case $\mathcal K_{\mathcal{A}} = \langle~\rangle $, the properties hold trivially, thus we assume that the synthesised controller is non-empty.

For~(i), trivially all dangling states are pruned, so it is always possible to reach a final state. Similarly, bad states (i.e.\ states in the set $R_s$) are never traversed by construction, i.e.\ transitions with target in $R_s$ are pruned.

For~(ii), by construction all uncontrollable transitions have source state in $R_s$, and thus are not reachable. Note that by Definition~\ref{def:choreographycontrollability} uncontrollable transitions are  necessary requirements that are not met and thus are always removed by the synthesis.
%

For~(iii), all transitions eventually violating strong safety are requests or offers and are pruned by the synthesis.

For~(iv), the transitions violating the branching condition are
\[\{\, (\vec q_1 \TRANS{\vec a}) = t \in T_{\mathcal K_{\mathcal A}} \mid \exists\,\vec{q_2} : (snd(\vec{a}) = i \wedge \ithel{\vec{q_1}}{i} = \ithel{\vec{q_2}}{i}) \wedge (\vec q_2 \TRANS{\vec a}) \not\in  T_{\mathcal K_{\mathcal A}} \,\}\]
and these are pruned by definition.
\end{proof}


Returning to Example~\ref{ex:semicontrollabilityandchoregraphy}, the erroneously accepted case is removed because, during the synthesis, the operation of pruning the transitions leading to  bad states causes the removal of the remaining transition. Thus, the obtained choreography is empty.
Similarly, the erroneously rejected case is not possible because, assuming that the output from the initial state is necessary, this necessary action is not rendered as uncontrollable as long as the output is matched by some other principal from the same initial state.

\modiff{We now estimate also the complexity of the choreography synthesis. With respect to the orchestration synthesis, in the choreography synthesis at each iteration a transition violating the branching condition can be removed. In the worst case, deciding if a transition violates the branching condition requires to check all other transitions.
Hence, an upper bound of the procedure for selecting a transition violating the branching condition is $\mathcal O(|T|^2)$.
Note that, in the unlikely event that all transitions share the same source state, the upper-bound complexity for computing the set of uncontrollable transitions is the same as in the case of orchestration synthesis.
Thus, a first upper bound of the complexity of the choreography synthesis algorithm is $\mathcal{O}((|T|+|Q|)\times|Q|^2\times|T|^4)$.
We can refine this first approximation to $\mathcal{O}((|T|+|Q|)\times|Q|^2)$. Indeed, similar to the case of orchestration synthesis, at each iteration in a single traversal of the automaton it is possible to compute the set of dangling states, the set of uncontrollable transitions, and the set of transitions violating the branching condition.
Also in this case, as for the other syntheses, our complexity estimation refers to the abstract specification of the algorithm resulting from Definition~\ref{def:synthesischoreography} and Theorem~\ref{the:synthesischoreography}. The  implementation of the algorithm could be optimised to perform even better.}


\begin{exa}\label{ex:choreography}
\modiff{We once more continue the running example by discussing the choreography synthesis of the running example for the composition
\[
\mathcal A_2 = \textsf{Client}\otimes \textsf{PrivilegedClient} \otimes \textsf{Broker}\otimes \textsf{Hotel} \otimes \textsf{Hotel}
\]
The choreography of $\mathcal A_2$ is depicted in Fig.~\ref{fig:chor} and the time needed to compute $\mathcal A_2$ and its choreography by using FMCAT is reported in Table~\ref{tab:results}.
Note that differently from $\mathcal A_1$ in Example~\ref{ex:mpc} and Example~\ref{ex:orchestration}, in $\mathcal A_2$ there is a privileged client and no privileged hotel. Indeed, \texttt{PrivilegedHotel} is not a valid contract for the choreography case.
The choreography does not need the overhead of interactions with the orchestrator and, most importantly, the synthesis of choreography does not introduce any additional behaviour.
Indeed, with due adjustment of necessary transitions of \texttt{PrivilegedClient} and \texttt{Hotel}, the choreography could be considered a sub-automaton of the orchestration.}

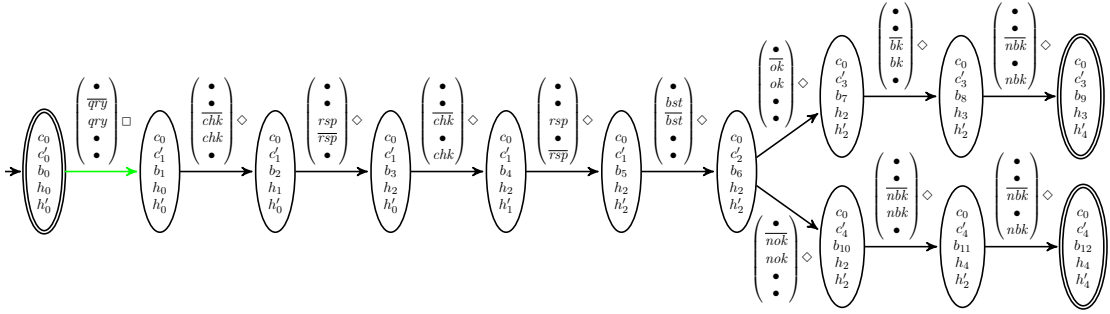
\begin{figure}
\centering
\begin{tikzpicture}[->,scale=.5,>=stealth',shorten >=1pt,auto,node distance=1cm,
        semithick, every node/.style={scale=0.5},initial text={},inner sep=0pt, minimum size=0pt]
\tikzstyle{every state}=[fill=white,draw=black,text=black,ellipse]

\node[initial,state,accepting] (0) {$\begin{array}{c}c_{0}\\ c'_{0}\\ b_{0}\\ h_{0}\\ h'_{0}\end{array}$};
\node[state] (1) [right = of 0] {$\begin{array}{c}c_{0}\\ c'_{1}\\ b_{1}\\ h_{0}\\ h'_{0}\end{array}$};
\node[state] (2) [right = of 1] {$\begin{array}{c}c_{0}\\ c'_{1}\\ b_{2}\\ h_{1}\\ h'_{0}\end{array}$};
\node[state] (3) [right = of 2] {$\begin{array}{c}c_{0}\\ c'_{1}\\ b_{3}\\ h_{2}\\ h'_{0}\end{array}$};
\node[state] (4) [right = of 3] {$\begin{array}{c}c_{0}\\ c'_{1}\\ b_{4}\\ h_{2}\\ h'_{1}\end{array}$};
\node[state] (5) [right = of 4] {$\begin{array}{c}c_{0}\\ c'_{1}\\ b_{5}\\ h_{2}\\ h'_{2}\end{array}$};
\node[state] (6) [right = of 5] {$\begin{array}{c}c_{0}\\ c'_{2}\\ b_{6}\\ h_{2}\\ h'_{2}\end{array}$};
\node[state] (7) [above right = -.15cm and 1cm of 6] {$\begin{array}{c}\,c_{0}\,\\ c'_{3}\\ b_{7}\\ h_{2}\\ h'_{2}\end{array}$};
\node[state] (8) [right = of 7] {$\begin{array}{c}\,c_{0}\,\\ c'_{3}\\ b_{8}\\ h_{3}\\ h'_{2}\end{array}$};
\node[state,accepting] (9) [right = of 8] {$\begin{array}{c}\,c_{0}\,\\ c'_{3}\\ b_{9}\\ h_{3}\\ h'_{4}\end{array}$};
\node[state] (10) [below right = -.15cm and 1cm of 6] {$\begin{array}{c}c_{0}\\ c'_{4}\\ b_{10}\\ h_{2}\\ h'_{2}\end{array}$};
\node[state] (11) [right = of 10] {$\begin{array}{c}c_{0}\\ c'_{4}\\ b_{11}\\ h_{4}\\ h'_{2}\end{array}$};
\node[state,accepting] (12) [right = of 11] {$\begin{array}{c}c_{0}\\ c'_{4}\\ b_{12}\\ h_{4}\\ h'_{4}\end{array}$};

\path
(0) edge[draw=green] node[above,yshift=5pt]{$\left(\!\!\begin{array}{c}\blk\\ \overline{\textit{qry}}\\ \textit{qry}\\ \blk\\ \blk\end{array}\!\!\right)\!\Necessary$} (1)
(1) edge[] node[above,yshift=5pt]{$\left(\!\!\begin{array}{c}\blk\\ \blk\\ \overline{\textit{chk}}\\ \textit{chk}\\ \blk\end{array}\!\!\right)\!\Permitted$} (2)
(2) edge[] node[above,yshift=5pt]{$\left(\!\!\begin{array}{c}\blk\\ \blk\\ \textit{rsp}\\ \overline{\textit{rsp}}\\ \blk\end{array}\!\!\right)\!\Permitted$} (3)
(3) edge[] node[above,yshift=5pt]{$\left(\!\!\begin{array}{c}\blk\\ \blk\\ \overline{\textit{chk}}\\ \blk\\ \textit{chk}\end{array}\!\!\right)\!\Permitted$} (4)
(4) edge[] node[above,yshift=5pt]{$\left(\!\!\begin{array}{c}\blk\\ \blk\\ \textit{rsp}\\ \blk\\ \overline{\textit{rsp}}\end{array}\!\!\right)\!\Permitted$} (5)
(5) edge[] node[above,yshift=5pt]{$\left(\!\!\begin{array}{c}\blk\\ \textit{bst}\\ \overline{\textit{bst}}\\ \blk\\ \blk\end{array}\!\!\right)\!\Permitted$} (6)
(6) edge[] node[above,xshift=-5pt,yshift=5pt]{$\left(\!\!\begin{array}{c}\blk\\ \overline{\textit{ok}}\\ \textit{ok}\\ \blk\\ \blk\end{array}\!\!\right)\!\Permitted$} (7)
(6) edge[] node[below,xshift=-5pt,yshift=-5pt]{$\left(\!\!\begin{array}{c}\blk\\ \overline{\textit{nok}}\\ \textit{nok}\\ \blk\\ \blk\end{array}\!\!\right)\!\Permitted$} (10)
(7) edge[] node[above,yshift=5pt]{$\left(\!\!\begin{array}{c}\blk\\ \blk\\ \overline{\textit{bk}}\\ \textit{bk}\\ \blk\end{array}\!\!\right)\!\Permitted$} (8)
(8) edge[] node[above,yshift=5pt]{$\left(\!\!\begin{array}{c}\blk\\ \blk\\ \overline{\textit{nbk}}\\ \blk\\ \textit{nbk}\end{array}\!\!\right)\!\Permitted$} (9)
(10) edge[] node[above,yshift=5pt]{$\left(\!\!\begin{array}{c}\blk\\ \blk\\ \overline{\textit{nbk}}\\ \textit{nbk}\\ \blk\end{array}\!\!\right)\!\Permitted$} (11)
(11) edge[] node[above,yshift=5pt]{$\left(\!\!\begin{array}{c}\blk\\ \blk\\ \overline{\textit{nbk}}\\ \blk\\ \textit{nbk}\end{array}\!\!\right)\!\Permitted$} (12);
\end{tikzpicture}
\caption{\label{fig:chor}\modiff{Choreography of 
$\textsf{Client}\otimes \textsf{PrivilegedClient} \otimes \textsf{Broker}\otimes \textsf{Hotel} \otimes \textsf{Hotel}$}}
\end{figure}


\begin{table}[t]
\centering
\def\arraystretch{1.2}
\color{black}
\begin{tabular}{p{5mm} | p{17.5mm} p{10mm} | p{17mm}  p{9mm} | p{19mm}  p{10mm} | p{19mm}  p{10mm}} 
& \shortstack[l]{\footnotesize Num.\ states \\ \footnotesize composition}
& \shortstack[l]{\footnotesize Time \\ \footnotesize (ms)}
& \shortstack[l]{\footnotesize Num.\ states \\ \footnotesize mpc \vphantom{ti}}
& \shortstack[l]{\footnotesize Time \\ \footnotesize (ms)}
& \shortstack[l]{\footnotesize Num.\ states \\ \footnotesize orchestration}
& \shortstack[l]{\footnotesize Time \\ \footnotesize (ms)}
& \shortstack[l]{\footnotesize Num.\ states \\ \footnotesize choreography}
& \shortstack[l]{\footnotesize Time \\ \footnotesize (ms)}
\\ \noalign{\hrule height 1pt}
$\mathcal A_1$ & \small 2934 & \small 65594 & \small 1 & \small 4070 & \small 37 & \small 715216 & \hspace{0.05mm}\small-- & \hspace{0.05mm}\small-- \\ 
$\mathcal A_2$ & \small 2934 & \small 66243 & \hspace{0.05mm}\small-- & \hspace{0.05mm}\small-- & \hspace{0.05mm}\small-- & \hspace{0.05mm}\small-- & \small 13 & \small 459311 
\end{tabular}
\def\arraystretch{1}
\caption[protect]{\label{tab:results}\color{black}{Results of computing the compositions $\mathcal A_1$ and $\mathcal A_2$ and their syntheses.}\footnotemark}
\end{table}

\footnotetext{\modiff{The evaluation was carried out on a machine with Processor Intel(R) Core(TM) i7-8500Y CPU at 1.50\,GHz, 1601\,Mhz, 2~Core(s), 4~Logical Processor(s) with 16\,GB of RAM, running 64-bit Windows~10.}} 

\modiff{We now use the example to discuss the differences between orchestration and choreography, and in particular the requirement that the \emph{branching condition} is satisfied.
In the orchestration, from the initial state either one of the two clients can interact. This decision is internally taken by the orchestrator (whose communications are abstracted away in the orchestration).
On the converse, in the choreography only the \texttt{PrivilegedClient} is allowed to interact.
This is because the clients are not able to decide on their own which one of them should start the interactions.
This can be explained as follows.
If both clients were allowed to interact, a deadlock could be reached upon the following steps. Initially, \texttt{PrivilegedClient} offers $\overline{\textit{qry}}$. Afterwards, for the interactions to continue such offer must be received by some principal (i.e.\ the underlining choreographed model is synchronous~\cite{BDFT15}).
In this case, $\texttt{Broker}$ receives the offer $\textit{qry}$.
At this point, \texttt{Client} is allowed to offer its $\overline{\textit{qry}}$ message.
The interactions are now deadlocked, because $\texttt{Broker}$  cannot receive such message, nor can any other contract.
This is an example of violation of the {branching condition}. Consider the initial state $\vec q_0 = (c_0,c'_0,b_0,h_0,h'_0)$ and state $\vec q_1 = (c_0,c'_1,b_1,h_0,h'_0)$.
In the orchestration, the {branching condition} is violated because from state $\vec q_0$ the match $(\overline{\textit{qry}}, \blk, \textit{qry}, \blk, \blk)$ is allowed, while it is not in state $\vec q_1$, and in both states $\texttt{Client}$  is in $c_0$.
During the choreography synthesis, the match $(\overline{\textit{qry}}, \blk, \textit{qry}, \blk, \blk)$ from state $\vec q_0$ is pruned.
Likewise, in the choreography the broker enquires the hotels in a fixed order, whereas in the orchestration all possible orders are allowed, or else the branching condition would be violated.}

\modiff{Note that an alternative choreography can be obtained by swapping the order in which the hotels are enquired by the broker.
Indeed, from state $\vec q_1$ the orchestration allows both matches $(\blk,\blk,\overline{\textit{chk}},\textit{chk},\blk)$ and
$(\blk,\blk,\overline{\textit{chk}}, \blk, \textit{chk})$. During the choreography synthesis, both these outgoing matches are violating the branching condition. By pruning one of them, the other is automatically amended, because the states causing violation of the branching condition become dangling. In particular, the synthesis prunes the transition $(\blk,\blk,\overline{\textit{chk}}, \blk, \textit{chk})$ in favour of $(\blk,\blk,\overline{\textit{chk}},\textit{chk},\blk)$.}

\modiff{Finally, concerning semi-controllability, note that it is not possible to have a choreography in which \texttt{PrivilegedClient} is not served in favour of \texttt{Client}, because the $\overline{\textit{qry}}$ offer of \texttt{PrivilegedClient} is \emph{necessary} and thus must be matched.}
\end{exa}




\section{Abstract Synthesis}\label{sect:discussion}

In Section~\ref{sect:background}, Section~\ref{sect:orchestrationsynthesis}, and Section~\ref{sect:choreographysynthesis}, we have presented three slightly different synthesis algorithms, and in the previous section we have illustrated their differences.
As said before, to bridge the gap between standard synthesis and  orchestration and choreography syntheses, the controllable and uncontrollable actions from SCT are related to permitted and necessary modalities, respectively, of MSCA\@.

The main intuition for this is that the SCT assumption of an unpredictable environment responsible for the uncontrollable transitions is not realistic in the case of coordination of services whose behaviour is known and observable.
As a result, necessary actions are not in correspondence with uncontrollable actions, but rather require the introduction of a milder notion of controllability.
The condition under which a controllable transition becomes uncontrollable varies depending on the particular synthesis algorithm (orchestration or choreography).
Conversely, in the standard mpc synthesis such information is local, i.e.\ a transition is declared to be uncontrollable.

In this section, we discuss an abstract synthesis algorithm that generalises the previous algorithms by abstracting away the conditions under which a transition is pruned or a state is deemed bad, thus encapsulating and extrapolating the notion of controllability and safety.
These two conditions, called \emph{pruning predicate} ($\phi_p$) and \emph{forbidden predicate} ($\phi_f$) are parameters to be instantiated by the corresponding instance of the synthesis algorithm (e.g.\ orchestration or choreography).
Predicate~$\phi_p$ is used for selecting the transitions to be pruned. Depending on the specific instance, non-local information about the automaton or the set of bad states is needed by $\phi_p$. Therefore, $\phi_p$ takes as input the current transition to be checked, the automaton, and the set of bad states. If $\phi_p$ evaluates to true, then the corresponding transition will be pruned.
Predicate~$\phi_f$ is used for deciding whether a state becomes bad. The input parameters are the same as $\phi_p$. However, $\phi_f$ only inspects necessary transitions ($T^\Box$). If $\phi_f$ evaluates to true, then the source state is deemed bad and added to the set $R_i$.
The abstract synthesis algorithm is formally defined below.

\begin{defi}[Abstract synthesis]\label{def:abstractsynthesis}
Let $\mathcal{A}$ be an MSCA, and let $\mathcal{K}_{0} = \mathcal A$ and $R_{0} = \textit{Dangling}(\mathcal{K}_{0})$.
Given two predicates $\phi_p, \phi_f: T \times MSCA \times Q \rightarrow \textsf{Bool}$,
we let the \emph{abstract synthesis function} $f_{(\phi_p, \phi_f)}: \textit{MSCA} \times  2^Q \rightarrow \textit{MSCA} \times  2^Q$ be defined as follows:
\begin{center}
\def\arraystretch{1.2}\begin{tabular}{c@{\hskip 0.5in}r@{\hskip 0.05in}c@{\hskip 0.05in}l}
\multicolumn{4}{l}{
$f_{(\phi_p, \phi_f)}(\mathcal{K}_{i-1},R_{i-1}) = (\mathcal{K}_{i},R_{i})$,
with}\\
& $T_{\mathcal{K}_{i}}$ & = & $T_{\mathcal{K}_{i-1}}\setminus
\{\, t \in T_{\mathcal K_{i-1}} \mid
                                        \phi_p(t, \mathcal K_{i-1}, R_{i-1}) = \textit{true}
                                        \,\}$\\
& $R_{i}$ & = & $R_{i-1} \cup
\{\,\vec q \mid (\vec q \TRANS{}) = t \in T_{\mathcal A}^\Box\,,\  \phi_f(t, \mathcal K_{i-1}, R_{i-1}) = \textit{true} \,\} \cup \textit{Dangling}(\mathcal{K}_{i})$
\end{tabular}
\end{center}
\end{defi}

As in the previous cases, the mpc relative to the pair $(\phi_p,\phi_f)$ is obtained by computing the least fixed point $(\mathcal K_{s}, R_{s})$ of $f_{(\phi_p,\phi_f)}$ and removing the states $R_{s}$ from $\mathcal K_{s}$.

\begin{restatable}[Abstract controller synthesis]{thm}{theabstractcontrollersynthesis}\label{the:abstractcontrollersynthesis}
The abstract synthesis function $f_{(\phi_p, \phi_f)}$ is monotone on the cpo $(\textit{MSCA} \times 2^Q, \leq)$ and its \emph{least fixed point} is:
\[
(\mathcal K^{(\phi_p, \phi_f)}_{s}, R^{(\phi_p, \phi_f)}_{s}) = \sup (\{\, f_{(\phi_p, \phi_f)}^n(\mathcal K_{0}, R_{0})\mid n \in \mathbb{N} \,\})
\]

\noindent
The \emph{abstract controller} of $\mathcal A$ for predicates ${(\phi_p, \phi_f)}$, denoted by $\mathcal K^{(\phi_p, \phi_f)}_{\mathcal A}$, is:
\[
\mathcal K^{(\phi_p, \phi_f)}_{\mathcal{A}} =
\left\{\begin{array}{l@{\qquad}l}
    \langle~\rangle
        & \mbox{if } \vec q_0 \in R^{(\phi_p, \phi_f)}_{s}\\
    \langle Q \setminus
    R^{(\phi_p, \phi_f)}_{s}, \vec q_0, A^{\Permitted}, A^{\Box}, A^{o},
        T_{\mathcal K^{(\phi_p, \phi_f)}_{s}}, F \setminus
        R^{(\phi_p, \phi_f)}_{s} \rangle
        & \mbox{otherwise}
\end{array}\right.
\]
\end{restatable}
\begin{proof}
The algorithm terminates because at each iteration either some transition is pruned or a state becomes forbidden, and both sets of transitions and states are finite.
We now prove that the synthesised automaton is (i)~non-blocking, (ii)~controllable, (iii)~most-permissive, and (iv)~safe.
In case $\mathcal K^{(\phi_p, \phi_f)}_{\mathcal{A}} = \langle~\rangle $, the properties hold trivially, thus we assume that the synthesised controller is non-empty.

For~(i), trivially all dangling states are pruned, so it is always possible to reach a final state. Similarly, bad states (i.e.\ states in the set $R_s$) are never traversed by construction, i.e.\ transitions with target in $R_s$ are pruned.

For~(ii) and~(iv), the forbidden predicate $\phi_f$ codifies exactly when controllability or safety is violated by a state. By construction, it is never the case that such a state is reached.

For~(iii), by construction a state is deemed bad or a transition is pruned exactly when either forbidden or pruning predicates are satisfied, respectively.  
Thus, maximality follows by the fact that each controller greater than the one synthesised will admit some forbidden state or transition. 
%
%
%
%
\end{proof}

%

In the remainder of this section, we show how to instantiate the abstract synthesis function to the standard synthesis function, to the orchestration synthesis function, or to the choreography synthesis function, and prove their correspondences.

\begin{restatable}[Abstract mpc synthesis]{thm}{theabstractmpcsynthesis}\label{the:abstractmpcsynthesis}
The standard synthesis function of Definition~\ref{def:mpc} coincides with the instantiation of the abstract synthesis function of Definition~\ref{def:abstractsynthesis} where, for a generic transition $t=(\vec q, \vec a, \vec q \, ')$, predicates $\phi_p$ and $\phi_f$ are defined as follows:
\begin{center}
\def\arraystretch{1.2}\begin{tabular}{r@{\hskip 0.05in}c@{\hskip 0.05in}l}
$\phi_p^{\textit{mpc}}(t, \mathcal K, R)$ & = & $(\vec  q\,' \in R) \vee (\vec q \text{ is forbidden})$\\
$\phi_f^{\textit{mpc}}(t, \mathcal K, R)$ & = & $(\vec  q\,' \in R)$
\end{tabular}
\end{center}
\end{restatable}
\begin{proof}
Let $\mathcal K_{\mathcal A}^{\textit{mpc}}$ and $\mathcal K_{\mathcal A}^{\textit{abs}}$ be the controllers computed through
Theorems~\ref{the:mpc} and~\ref{the:abstractcontrollersynthesis}, respectively.
The proof proceeds by induction on the fixed point iterations and by case analysis.

For the base case,
by definition $\mathcal K_{0}^{\textit{mpc}} = \mathcal K_{0}^{\textit{abs}} = \mathcal A$ and
$R_{\,0}^{\textit{abs}} = R_{\,0}^{\textit{mpc}} = \textit{Dangling}(\mathcal K_{0})$.

For the inductive case, let $i$ be a fixed point iteration. Assuming
$\mathcal K_{i-1}^{\textit{mpc}} = \mathcal K_{i-1}^{\textit{abs}}$ and $R_{\,i-1}^{\textit{mpc}} = R_{\,i-1}^{\textit{abs}}$,
we prove
$\mathcal K_{i}^{\textit{mpc}}  = \mathcal K_{i}^{\textit{abs}}$ and $R_{\,i}^{\textit{mpc}} = R_{\,i}^{\textit{abs}}$.

The equivalence $\mathcal K_{i}^{\textit{mpc}}  = \mathcal K_{i}^{\textit{abs}}$ follows because at the $i$th iteration, $\phi^{\textit{mpc}}_p$ detects exactly the same transitions that are pruned by the mpc synthesis algorithm.

For the equivalence $R_{\,i}^{\textit{mpc}}= R_{\,i}^{\textit{abs}}$, we have
$R_{\,i}^{\textit{mpc}} = R_{\,i-1}^{\textit{mpc}} \cup \textit{Dangling}(\mathcal K_{i}^{\textit{mpc}}) \cup
\{\, \vec q \mid (\vec q, \vec a, \vec q \, ')\in T_{\mathcal{K}_{i}^{\textit{mpc}}}^\Box,\ \vec q \, '\in R_{\,i-1}^{\textit{mpc}} \,\}$ and
$R_{\,i}^{\textit{abs}} = R_{\,i-1}^{\textit{abs}} \cup \textit{Dangling}(\mathcal K_{i}^{\textit{abs}}) \cup
\{\, \vec q \mid ( \vec q \TRANS{ a}) = t \in T_{\mathcal A}^\Box,\linebreak  \phi^{\textit{mpc}}_f(t, \mathcal K_{i-1}^{\textit{abs}}, R_{\,i-1}^{\textit{abs}}) = \textit{true}\}$.

Since $\mathcal K_{\,i}^{\textit{mpc}}  = \mathcal K_{\,i}^{\textit{abs}}$, also the dangling states are equivalent.
It remains to prove that
\[
\{\, \vec q \mid (\vec q, \vec a, \vec q \, ')\in T_{\mathcal{K}_{\,i}^{\textit{mpc}}}^\Box,\ \vec q \, '\in R_{\,i-1}^{\textit{mpc}} \,\}
=
\{\, \vec q \mid ( \vec q \TRANS{ a}) = t \in T_{\mathcal A}^\Box,\  \phi^{\textit{mpc}}_f(t, \mathcal K_{i-1}^{\textit{abs}}, R_{\,i-1}^{\textit{abs}}) = \textit{true}\}
.\]
This equivalence is straightforward by the definition of $\phi^{\textit{mpc}}_f$ and the inductive hypothesis.
\end{proof}

Note that in Theorem~\ref{the:abstractmpcsynthesis} the predicates do not use any non-local information related to the parameter $\mathcal K$.
For both orchestration and choreography, two different semi-controllability conditions are used to decide whether a state has become forbidden.
These conditions are translated into the corresponding forbidden predicates.

\begin{restatable}[Abstract orchestration synthesis]{thm}{theabstractorchestrationsynthesis}\label{the:abstractorchestrationsynthesis}
The orchestration synthesis function of Definition~\ref{def:synthesisorchestration} coincides with the instantiation of the abstract synthesis function of Definition~\ref{def:abstractsynthesis} where, for a generic transition $t=(\vec q, \vec a, \vec q \, ')$, predicates $\phi_p$ and $\phi_f$ are defined as follows:
\begin{center}
\def\arraystretch{1.2}\begin{tabular}{r@{\hskip 0.05in}c@{\hskip 0.05in}l}
$\phi_p^{\textit{orc}}(t, \mathcal K, R)$ & = & $(t \text{ is a request }) \vee (\vec  q\,' \in R)$\\
$\phi_f^{\textit{orc}}(t, \mathcal K, R)$ & = & $\nexists\,(\vec{q_2} \TRANS{\vec a_2} \vec{q_2}\,\!') \in T^\Box_{\mathcal K} \,:\,
 (\vec a_2 \text{ is a match })
 \wedge
 (\vec q_2, \vec q_2\!' \not\in \textit{Dangling}(\mathcal K))$\\
& & $\phantom{\nexists\,(\vec{q_2} \TRANS{\vec a_2} \vec{q_2}\,\!') \in T^\Box_{\mathcal K} \,:\,}
\wedge
(\ithel{\vec{q}}{i} = \ithel{\vec{q_2}}{i})
\wedge
(\ithel{\vec{a}}{i} = \ithel{\vec{a_2}}{i} = a)$
\end{tabular}
\end{center}
\end{restatable}
\begin{proof}[Proof (sketch)]
The proof is analogous to that of Theorem~\ref{the:abstractmpcsynthesis} but relying on  Theorem~\ref{def:synthesisorchestration} instead of Theorem~\ref{the:mpc}. The full proof can be found in the appendix.
\end{proof}
\iftoggle{APPENDIX}
{}
{
\begin{proof}
The proof is analogous to that of Theorem~\ref{the:abstractmpcsynthesis} but relying on  Theorem~\ref{def:synthesisorchestration} instead of Theorem~\ref{the:mpc}. The full proof follows.

Let $\mathcal K_{\mathcal A}^{\textit{orc}}$ and $\mathcal K_{\mathcal A}^{\textit{abs}}$ be the controllers computed through
Theorems~\ref{def:synthesisorchestration} and~\ref{the:abstractcontrollersynthesis}, respectively.
The proof proceeds by induction on the fixed point iterations and by case analysis.

For the base case,
by definition $\mathcal K_{0}^{\textit{orc}} = \mathcal K_{0}^{\textit{abs}} = \mathcal A$ and
$R_{0}^{\textit{abs}} = R_{0}^{\textit{orc}} = \textit{Dangling}(\mathcal K_0)$.

For the inductive case, let $i$ be a fixed point iteration. Assuming
$\mathcal K_{i-1}^{\textit{orc}} = \mathcal K_{i-1}^{\textit{abs}}$ and $R_{i-1}^{\textit{orc}}= R_{i-1}^{\textit{abs}}$,
we prove
$\mathcal K_{i}^{\textit{orc}}  = \mathcal K_{i}^{\textit{abs}}$ and $R_{i}^{\textit{orc}}= R_{i}^{\textit{abs}}$.

The equivalence $\mathcal K_{i}^{\textit{orc}}  = \mathcal K_{i}^{\textit{abs}}$ follows because at the $i$th iteration $\phi^{\textit{orc}}_p$ detects exactly the same transitions that are pruned by the orchestration synthesis algorithm.

For the equivalence $R_{i}^{\textit{orc}}= R_{i}^{\textit{abs}}$, we have
$R_{i}^{\textit{orc}} = R_{\,i-1}^{\textit{orc}}
 \cup
\{\,\vec q \mid (\vec q \TRANS{})\in T^\Box_{\mathcal A} \textit{ is un}\-\textit{controllable in } \mathcal{K}_{\,i}^{\textit{orc}}\,\}
\cup \textit{Dangling}(\mathcal{K}_{\,i}^{\textit{orc}})$, and
$R_{i}^{\textit{abs}} = R_{\,i-1}^{\textit{abs}}
\cup
\textit{Dangling}(\mathcal K_i^{\textit{abs}}) \cup
\{\, \vec q \mid (\vec q\TRANS{a}) = t \in T_{\mathcal A}^\Box,\  \phi^{\textit{orc}}_f(t, \mathcal K_{i-1}^{\textit{abs}}, R_{i-1}^{\textit{abs}}) = \textit{true}\}$.

Since $\mathcal K_{i}^{\textit{orc}}  = \mathcal K_{i}^{\textit{abs}}$, also the dangling states are equivalent.
It remains to prove that
$
\{\,\vec q \mid (\vec q \TRANS{})\in T^\Box_{\mathcal A} \textit{ is uncontrollable in } \mathcal{K}_{\,i}^{\textit{orc}}\,\}
=
\{\, \vec q \mid ( \vec q \TRANS{}) = t \in T_{\mathcal A}^\Box,\  \phi^{\textit{orc}}_f(t, \mathcal K_{i-1}^{\textit{abs}}, R_{i-1}^{\textit{abs}}) = \textit{true}\}
$.
This equivalence is straightforward by the definition of $\phi^{\textit{orc}}_f$, Definition~\ref{def:controllabilityorchestration}, and the inductive hypothesis.
\end{proof}

}

The pruning predicate of Theorem~\ref{the:abstractorchestrationsynthesis} does not use any information coming from the global automaton $\mathcal K$, whereas this is no longer the case for the forbidden predicate that indeed specifies the semi-controllability condition for the necessary transitions of an orchestration (cf.\  Definition~\ref{def:controllabilityorchestration}).

\begin{restatable}[Abstract choreography synthesis]{thm}{theabstractchoreographysynthesis}\label{the:abstractchoreographysynthesis}
The choreography synthesis function of Definition~\ref{def:synthesischoreography} coincides with the instantiation of the abstract synthesis function of Definition~\ref{def:abstractsynthesis}, where given a generic transition $t=(\vec q, \vec a, \vec q \, ')$, the predicates $\phi_p$ and $\phi_f$ are defined as follows, where $\hat{T}_{\mathcal K, R}$ is defined in  Definition~\ref{def:synthesischoreography}:
\begin{center}
\def\arraystretch{1.2}\begin{tabular}{r@{\hskip 0.05in}c@{\hskip 0.05in}l}
$\phi_p^{\textit{cor}}(t, \mathcal K, R)$ & = & $(t \text{ is a request or offer}) \vee (\vec  q\,' \in R) \vee t \in \hat{T}_{\mathcal K, R}$\\
$\phi_f^{\textit{cor}}(t, \mathcal K, R)$ & = & $\nexists\,(\vec{q} \TRANS{\vec a_2} \vec{q_2}') \in T^\Box_{\mathcal K} \,:\,
(\vec a_2 \text{ is a match})
\wedge
(\vec q, \vec q_2\!' \not\in \textit{Dangling}(\mathcal K) )$\\
& & $\phantom{\nexists\,(\vec{q} \TRANS{\vec a_2} \vec{q_2}') \in T^\Box_{\mathcal K} \,:\,}
\wedge
(\ithel{\vec{a}}{i} = \ithel{\vec{a_2}}{i} = \overline a)$
\end{tabular}
\end{center}
\end{restatable}
\begin{proof}[Proof (sketch)]
The proof is analogous to that of Theorem~\ref{the:abstractmpcsynthesis} but relying on  Theorem~\ref{def:synthesischoreography} instead of Theorem~\ref{the:mpc}. The full proof can be found in the appendix.
\end{proof}
\iftoggle{APPENDIX}
{}
{
\begin{proof}
The proof is analogous to that of Theorem~\ref{the:abstractmpcsynthesis} but relying on  Theorem~\ref{def:synthesischoreography} instead of Theorem~\ref{the:mpc}. The full proof follows.

Let $\mathcal K_{\mathcal A}^{\textit{cor}}$ and $\mathcal K_{\mathcal A}^{\textit{abs}}$ be the controllers computed through
Theorems~\ref{def:synthesischoreography} and~\ref{the:abstractcontrollersynthesis}, respectively.
The proof proceeds by induction on the fixed point iterations and by case analysis.

For the base case,
by definition $\mathcal K_{0}^{\textit{cor}} = \mathcal K_{0}^{\textit{abs}} = \mathcal A$ and
$R_{0}^{\textit{abs}} = R_{0}^{\textit{cor}} = \textit{Dangling}(\mathcal K_0)$.

For the inductive case, let $i$ be a fixed point iteration. Assuming
$\mathcal K_{i-1}^{\textit{cor}} = \mathcal K_{i-1}^{\textit{abs}}$ and $R_{i-1}^{\textit{cor}}= R_{i-1}^{\textit{abs}}$,
we prove
$\mathcal K_{i}^{\textit{cor}}  = \mathcal K_{i}^{\textit{abs}}$ and $R_{i}^{\textit{cor}} = R_{i}^{\textit{abs}}$.

The equivalence $\mathcal K_{i}^{\textit{cor}} = \mathcal K_{i}^{\textit{abs}}$ follows because at the $i$th iteration $\phi^{\textit{cor}}_p$ detects exactly the same transitions that are pruned by the choreography synthesis algorithm (and takes the same non-deterministic choices).

For the equivalence $R_{i}^{\textit{cor}}= R_{i}^{\textit{abs}}$, we have $R_{i}^{\textit{cor}} = R_{\,i-1}^{\textit{cor}}
 \cup
\{\,\vec q \mid (\vec q \TRANS{})\in T^\Box_{\mathcal A} \textit{ is un}\-\textit{controllable in } \mathcal{K}_{\,i}^{\textit{cor}}\,\}
\cup \textit{Dangling}(\mathcal{K}_{\,i}^{\textit{cor}})$, and
$R_{i}^{\textit{abs}} = R_{\,i-1}^{\textit{abs}}
\cup
\textit{Dangling}(\mathcal K_i^{\textit{abs}}) \cup
\{\, \vec q \mid (\vec q\TRANS{a}) = t \in T_{\mathcal A}^\Box,\  \phi^{\textit{cor}}_f(t, \mathcal K_{i-1}^{\textit{abs}}, R_{i-1}^{\textit{abs}}) = \textit{true}\}$.

Since $\mathcal K_{i}^{\textit{cor}} = \mathcal K_{i}^{\textit{abs}}$, also the dangling states are equivalent.
It remains to prove that
$
\{\,\vec q \mid (\vec q \TRANS{})\in T^\Box_{\mathcal A} \textit{ is uncontrollable in } \mathcal{K}_{\,i}^{\textit{cor}}\,\}
=
\{\, \vec q \mid ( \vec q \TRANS{}) = t \in T_{\mathcal A}^\Box,\  \phi^{\textit{cor}}_f(t, \mathcal K_{i-1}^{\textit{abs}}, R_{i-1}^{\textit{abs}}) = \textit{true}\}
$.
This equivalence is straightforward by the definition of $\phi^{\textit{cor}}_f$, Definition~\ref{def:choreographycontrollability}, and the inductive hypothesis.
\end{proof}

}

Notably, in Theorem~\ref{the:abstractchoreographysynthesis} both predicates require global information on the whole automaton. Similarly to Theorem~\ref{the:abstractorchestrationsynthesis}, the forbidden predicate codifies the semi-controllability condition of Definition~\ref{def:choreographycontrollability}.
Moreover, the pruning predicate removes all transitions violating the branching condition (cf.\ Definition~\ref{def:branchingcondition}).


\section{A Partial Order on Controllers}\label{sect:po}

In Theorem~\ref{the:abstractmpcsynthesis}, Theorem~\ref{the:abstractorchestrationsynthesis}, and Theorem~\ref{the:abstractchoreographysynthesis}, we have proved that the three previously presented synthesis algorithms are instantiations of the abstract synthesis algorithm of Definition~\ref{def:abstractsynthesis}.
This abstraction provides us the mean to formally relate the various algorithms presented so far, as detailed in this section.

To begin with, we define a partial order on predicates.
Intuitively, a pair $(\phi_{p_2}, \phi_{f_{2}})$ is greater than another pair $(\phi_{p_1}, \phi_{f_{1}})$ if and only if $(\phi_{p_2}, \phi_{f_{2}})$  is (pairwise) entailed by $(\phi_{p_1}, \phi_{f_{1}})$.

\begin{defi}[Partial order on predicates]\label{def:popredicates}
Let $\mathcal A$ be an \acronym\ and let $\prset$ be the set of pairs of pruning and forbidden predicates of Definition~\ref{def:abstractsynthesis} with $(\phi_{p_1}, \phi_{f_{1}}), (\phi_{p_2}, \phi_{f_{2}}) \in \prset$. 
The partial order on predicates $(\prset, \leq)$ is defined as:
\begin{center}
\def\arraystretch{1.2}\begin{tabular}{r@{\hskip 0.05in}c@{\hskip 0.05in}l}
$(\phi_{p_1}, \phi_{f_{1}}) \leq (\phi_{p_2}, \phi_{f_{2}})$ & \text{iff} & $\forall i \in \mathbb N\ .\ (\phi_{p_1}(t,\mathcal K^1_i, R^1_i) \Rightarrow (\phi_{p_2}(t,\mathcal K^2_i, R^2_i) \vee t \not\in \mathcal K^2_i))$\\
& & $\hspace{10.2mm}{} \wedge (\phi_{f_{1}}(t,\mathcal K^1_i, R^1_i) \Rightarrow (\phi_{f_{2}}(t,\mathcal K^2_i, R^2_i))\vee \vec q\in\textit{Dangling}(\mathcal K^2_i))$,
\end{tabular}
\end{center}
where $t=(\vec q, \vec a, \vec q\,')$.
\end{defi}

By Definition~\ref{def:abstractsynthesis}, we know that such predicates are used to refine an \acronym\ during the synthesis.
Indeed, states and transitions are removed when such predicates are satisfied by them.
The partial order on predicates induces an ordering on the various abstract controllers, as the following result shows.

\begin{restatable}[Ordering controllers]{proposition}{propcontrollerorder}\label{prop:controllerorder}
Let $\mathcal A$ be an \acronym\ and let $(\phi_{p_1}, \phi_{f_{1}}), (\phi_{p_2}, \phi_{f_{2}}) \in \prset$ be such that $(\phi_{p_1}, \phi_{f_{1}}) \leq (\phi_{p_2}, \phi_{f_{2}})$. Then:
\[
\mathcal K^{(\phi_{p_2}, \phi_{f_{2}})}_{\mathcal A} \subseteq \mathcal K^{(\phi_{p_1}, \phi_{f_{1}})}_{\mathcal A}
\]
\end{restatable}
\begin{proof}
By Definition~\ref{def:abstractsynthesis}, both
$\mathcal K^{(\phi_{p_1}, \phi_{f_{1}})}_{\mathcal A}$ and $\mathcal K^{(\phi_{p_2}, \phi_{f_{2}})}_{\mathcal A}$ are sub-automata of $\mathcal A$, and they only differ in the sets of states and transitions.

By contradiction, assume that there exists a transition $t$ in
$
T_{\mathcal K^{(\phi_{p_2}, \phi_{f_2})}_{\mathcal A}}
\setminus
T_{\mathcal K^{(\phi_{p_1}, \phi_{f_1})}_{\mathcal A}}
$.
By Definition~\ref{def:abstractsynthesis}, let $i$ be the iteration where $t$ is removed from $T_{\mathcal K^{1}_{i}}$.
By hypothesis, it holds that
$\phi_{p_1}(t,\mathcal K^1_i, R^1_i) \Rightarrow \phi_{p_2}(t,\mathcal K^2_i, R^2_i) \vee t \not\in \mathcal K^2_i$,
hence by Definition~\ref{def:abstractsynthesis}, $t$ must also have been removed from
$T_{\mathcal K^{2}_{i}}$ or it is not present, a contradiction.

Similarly,
assume that there exists a state $\vec q$  in
$
Q_{\mathcal K^{(\phi_{p_2}, \phi_{f_{2}})}_{\mathcal A}}
\setminus
Q_{\mathcal K^{(\phi_{p_1}, \phi_{f_{1}})}_{\mathcal A}}
$.
By Definition~\ref{def:abstractsynthesis}, let $i$ be the iteration where $\vec q$ is added to $R_{\mathcal K^{1}_{i}}$.
By hypothesis, it holds that
$\phi_{f_1}(t,\mathcal K^1_i, R^1_i) \Rightarrow \phi_{f_2}(t,\mathcal K^2_i, R^2_i) \vee \vec q\in\textit{Dangling}(\mathcal K^2_i)$,
hence by Definition~\ref{def:abstractsynthesis}, $\vec q$ must also have been added to
$R_{\mathcal K^{2}_{i}}$.
Finally,  $Q_{\mathcal K^{(\phi_{p_2}, \phi_{f_{2}})}_{\mathcal A}} = Q_{\mathcal A} \setminus R^{2}_s$ and
$R_{\mathcal K^{2}_{i}} \subseteq R^{2}_s$, thus a contradiction is reached.
\end{proof}



This result has an immediate application in performing abstraction of syntheses, in the sense that
the lesser the pair of predicates the more abstract (in refinement terms) the corresponding synthesised automaton.
This can be useful to perform partial syntheses and skip unnecessary checks or even potentially undecidable computations.
For example, if $\mathcal K^{(\phi_{p_1},\,\phi_{f_{1}})} = \langle~\rangle$, for a given pair $(\phi_{p_1}, \phi_{f_{1}})$, then by Proposition~\ref{prop:controllerorder} we know that
for all $(\phi_{p_i}, \phi_{f_{i}})$ such that $(\phi_{p_1}, \phi_{f_{1}}) \leq (\phi_{p_i}, \phi_{f_{i}})$ it will hold that $\mathcal K^{(\phi_{p_i},\,\phi_{f_{i}})} = \langle~\rangle$.

While the orchestration synthesis of Definition~\ref{def:synthesisorchestration} is enforcing agreement, the mpc synthesis of Definition~\ref{def:mpc} is enforcing a generic predicate modelled as forbidden states.
Whenever the mpc synthesis is also enforcing agreement, as an instantiation of Proposition~\ref{prop:controllerorder}, we can prove that the two syntheses are related.
Moreover, agreement identifies forbidden transitions as those labelled by requests.
On the converse, the mpc synthesis identifies forbidden states rather than forbidden transitions.
Therefore, to enable a comparison of the mpc and the orchestration synthesis,
we need to
(i)~transform the automaton such that the predicate on forbidden transitions  (i.e.\ agreement in this case) can be expressed by means of forbidden states and
(ii)~instantiate the generic predicate expressed by forbidden states.
For point~(i), the synthesis of the mpc is applied to the automaton ${\mathcal A}'$ obtained from the original automaton ${\mathcal A}$ by erasing controllable forbidden transitions.
For point~(ii), forbidden states are those states that are sources of uncontrollable forbidden transitions.
This is what the following lemma states.

\begin{restatable}[Orchestration vs.\ mpc synthesis]{lemma}{lemmpcincludedorc}%
\label{lem:mpcincludedorc}
Given an \acronym\ $\mathcal A$,
let $\mathcal A'$ be obtained from $\mathcal A$ by removing all controllable request transitions and considering as forbidden the  states of $\mathcal A'$  with outgoing uncontrollable request transitions.
Let $\mathcal K_\mathcal{A}^{\textit{orc}}$ and $\mathcal K_{\mathcal{A}'}^{\textit{mpc}}$ be the orchestration and mpc of Definitions~\ref{def:synthesisorchestration} and~\ref{def:mpc}, respectively.
Then:
\[
\mathcal K_{\mathcal{A}'}^{\textit{mpc}}\subseteq\mathcal K_\mathcal{A}^{\textit{orc}}
\]
\end{restatable}
\begin{proof}
By Theorems~\ref{the:abstractorchestrationsynthesis} and~\ref{the:abstractmpcsynthesis}, $\mathcal K_\mathcal{A}^{\textit{orc}}$ and $\mathcal K_{\mathcal{A}'}^{\textit{mpc}}$ are equivalent to $\mathcal K_\mathcal{A}^{(\phi_{p}^{\textit{orc}},\,\phi_{f}^{\textit{orc}})}$ and $\mathcal K_{\mathcal{A}'}^{(\phi_{p}^{\textit{mpc}},\,\phi_{f}^{\textit{mpc}})}$, respectively.
Moreover, both controllers are sub-automata of $\mathcal A$, and they only differ in the sets of states and transitions.
%
%

Recall that, given $t=(\vec q, \vec a, \vec q\,')$,
$
\phi_p^{\textit{mpc}}(t, \mathcal K^{\textit{mpc}}_i, R^{\textit{mpc}}_i) =  (\vec q\,' \in R^{\textit{mpc}}_i) \vee (\vec q \text{ is forbidden})
$,
$
\phi_f^{\textit{mpc}}(t, \mathcal K^{\textit{mpc}}_i, R^{\textit{mpc}}_i) = (\vec q\,' \in R^{\textit{mpc}}_i)
$
and
$
 \phi_p^{\textit{orc}}(t, \mathcal K^{\textit{orc}}_i, R^{\textit{orc}}_i) = (t \text{ is a request }) \vee (\vec q\,' \in R^{\textit{orc}}_i)
$,
$
 \phi_f^{\textit{orc}}(t, \mathcal K^{\textit{orc}}_i, R^{\textit{orc}}_i) = \nexists\,(\vec{q_2} \TRANS{\vec a_2} \vec{q_2}') \in T^\Box_{\mathcal K^{\textit{orc}}_i} \,:\,
 (\vec a_2 \text{ is a match })
 \wedge
 (\vec q_2, \vec{q_2}' \not\in \textit{Dangling}(\mathcal K^{\textit{orc}}_i))
\wedge
(\ithel{\vec{q}}{i} = \ithel{\vec{q_2}}{i})
\wedge
(\ithel{\vec{a}}{i} = \ithel{\vec{a_2}}{i} = a)
$.

We proceed by induction on $i$.
For the base case, it holds that $\mathcal K_{{\mathcal A}'_0} \subseteq \mathcal K_0$ and
$\textit{Dangling}(\mathcal K_0) \subseteq \textit{Dangling}(\mathcal K_{{\mathcal A}'_0})$.
By hypothesis, $\phi_p^{\textit{orc}}(t, \mathcal K^{\textit{orc}}_0, R^{\textit{orc}}_0)$ is true. Then either $t$ is a request or $\vec q\,' \in \textit{Dangling}(\mathcal K_0)$.
If $t$ is a request, then $t$ has been already pruned.
Otherwise, $\vec q\,' \in \textit{Dangling}(\mathcal K_0)$ (or both), and so it is in $\textit{Dangling}(\mathcal K_{{\mathcal A}'_0})$ and the pruning predicate of the mpc is satisfied.
Similarly, by hypothesis $\phi_f^{\textit{orc}}(t, \mathcal K^{\textit{orc}}_0, R^{\textit{orc}}_0)$ is true. Since no transitions have been pruned in
$ \mathcal K_0$, it must be the case that the source state of $t$ is in $\textit{Dangling}(\mathcal K_0)$, and so it is in $\textit{Dangling}(\mathcal K_{{\mathcal A}'_0})$.

For the inductive step, the implication on the pruning predicate is satisfied by noticing that $R_{i-1}^{\textit{orc}} \subseteq R_{i-1}^{\textit{mpc}}$.
The implication on the forbidden predicate is satisfied because trivially $t \not\in T^\Box_{\mathcal K^{\textit{orc}}_i}$, and hence
$t \not\in T^\Box_{\mathcal K^{\textit{mpc}}_i}$, and this is because either the target is dangling or the source is forbidden.
In both cases the forbidden predicate of the mpc is satisfied.
\end{proof}



Thus, for example, given an MSCA $\mathcal A$, from $\mathcal K_\mathcal{A}^{\textit{orc}} = \langle~\rangle$ we can conclude that $\mathcal K_\mathcal{A}^{\textit{mpc}} = \langle~\rangle$ by Lemma~\ref{lem:mpcincludedorc}, without actually computing it.



\begin{exa}
\modiff{Concluding the running example, one can observe that the mpc of $\mathcal A_1$ is a sub-automaton (formed of only the initial and final state) of the orchestration of $\mathcal A_1$.}
\end{exa}



\section{Related Work}\label{sect:related}

Our contributions to bridging the gap between SCT and coordination of services concern adaptations of the classical synthesis algorithm from SCT in order to synthesise orchestrations and choreographies of service contracts formalised as MSCA\@.
In the literature, there exist many formalisms for modelling and analysing (service) contracts, ranging from behavioural type systems, including behavioural contracts~\cite{CGP09,ABZ13,LP15} and session types~\cite{BLMT08,HYC08,DD09,CDP12,MNF13}, to automata-based formalisms, including interface automata~\cite{AH01} and (timed) (I/O) automata~\cite{LT89,AD94,DLLNW10}. Foundational models for service contracts and session types are surveyed in~\cite{BBG07,BCZ15,Hut16}.

The MSCA formalism used in this paper differs fundamentally from these models, which typically study notions of contract compliance involving only two parties, since MSCA primitively support \emph{multi-party} compliance of contracts that \emph{compete} on offering or requesting the same service.
Furthermore, the above models do not consider modalities of services whereas MSCA provide primitive support for \emph{permitted} and \emph{necessary} service actions, which resulted in the introduction of a novel notion of \emph{semi-controllability} in the context of SCT\@.
Modal Transition Systems (MTS) and their extensions~\cite{Kre17}, as adopted for instance in Software Product Line Engineering (SPLE~\cite{PBL05,ABKS13}), like modal I/O automata~\cite{LNW07} and MTS with variability constraints~\cite{BFGM16}, do natively distinguish may and must modalities, but the other differences remain. In particular, they cannot explicitly handle dynamic composition by allowing new services that join composite services to intercept already matched actions.

We are only aware of two other applications of SCT to MTS\@. In~\cite{DDM10}, there is no direct relation between may/must and controllable/uncontrollable, and the modal automaton (i.e.\ MTS with final states) is seen as a predicate  that is satisfied if the plant automaton (i.e.\ the system to be refined against the predicate) is a sort of alternate refinement of the predicate.
Similarly, in~\cite{FP07}, the control objectives (i.e.\ the predicate) is a modal automaton, non-blockingness is not considered, and another modal automaton describes which actions are controllable and which are uncontrollable in the plant automaton.
In this paper, the predicate is an invariant (i.e.\ forbidden states and forbidden transitions are given), the modal automaton (i.e.\ MSCA) is the plant, and a necessary transition induces different notions of controllability according to the adopted coordination paradigm.

SCT was first applied to SPLE in~\cite{BRV16} by showing how the \texttt{CIF\,3} toolset~\cite{vBFHHMvdMR14} can automatically synthesise a single (global, family) model representing an automaton for each of the valid products of a product line from (i)~a feature constraint with attributes (e.g.\ cost), (ii)~behavioural component models associated with the features, and (iii)~additional behavioural requirements like state invariants, action orderings, and guards on actions (reminiscent of the Featured Transition Systems of~\cite{CCSHLR13}). The resulting \texttt{CIF\,3} model satisfies all feature-related constraints as well as all given behavioural requirements. Since \texttt{CIF\,3} allows the export of such models in a format accepted by the \texttt{mCRL2} model checker~\cite{CGKSVWW13}, the latter can be used to verify arbitrary behavioural properties expressed in the modal $\mu$-calculus with data or its feature-oriented variant of~\cite{BVW17}. An important advantage is that both \texttt{CIF\,3} and \texttt{mCRL2} can be used off-the-shelf, meaning that no additional tools are required. Differently from our approach, all actions are controllable and orchestration is not considered. In~\cite{BDFT16}, the prototypical tool CAT supporting orchestration synthesis for CA is presented.

The only approach by others to bridge the gap between SCT and coordination of services that we are aware of is that of~\cite{ADR16}, where services are formalised as so-called Service Labelled Transition Systems (SLTS), which are a kind of guarded automata with data.
To this aim, SCT is adapted to deal with conditions and variables as well as with a means to enforce services based on runtime information.
However, service composition through SLTS is based on the standard synchronous product, whilst the contract composition expresses competing contracts.
More importantly, in~\cite{ADR16}, input actions are considered uncontrollable whereas output actions are controllable, in the standard view of a service interacting with the environment.
Our contribution induces novel notions of controllability to express  necessary requirements that are semi-controllable.
The standard controller synthesis algorithm is used in~\cite{GMW12} to synthesise adapters between services. These adapters act like proxies and are used to enforce properties such as deadlock-freedom. Compared to our work, the interactions
between services are driven by their contracts rather than by adapters. The standard controller synthesis algorithm cannot be applied to synthesise a correct composition of contracts.

We conclude this section by describing two recent extensions of MSCA, developed for different purposes, and for which we also defined adapted synthesis algorithms.
In~\cite{JSCP}, we presented Featured Modal Contract Automata (FMCA). Technically, we extended MSCA with a variability mechanism concerning structural constraints that operate on the service contract, used to define different configurations. This reflects the fact that services are typically reused in configurations that vary over time and need to dynamically adapt to changing environments~\cite{survey}. Configurations were characterised by which service actions are mandatory and which forbidden. The valid configurations were defined as those respecting all structural constraints. We followed the well-established paradigm of SPLE, which aims at efficiently managing a product line (family) of highly (re)configurable systems to allow for mass customisation~\cite{PBL05,ABKS13}. To compactly represent a product line, i.e.\ the set of valid product configurations, we used a so-called feature constraint, a propositional formula $\varphi$ whose atoms are features~\cite{Bat05}, and we identified features as service actions (offers as well as requests). A valid product then distinguishes a set of mandatory and a set of forbidden actions. 
Consequently, we defined an algorithm to compute the FMCA $\mathcal K_{\mathcal{A}_{p}}$ as the mpc for a valid product~$p$ of an FMCA $\mathcal A$.
The main adaptation of the synthesis algorithm for MSCA was to consider as bad states also those that cannot prevent a forbidden action to be eventually executed and to discard the transitions labelled with actions forbidden by~$p$. Moreover, if some action that is mandatory in~$p$ is unavailable in the automaton that results from the fixed point iteration, then the mpc results empty.
\modiff{In~\cite{JSCP}, we also presented an evaluation of FMCA with the prototypical tool FMCAT\@. Building on CAT~\cite{BDFT16}, FMCAT can synthesise the orchestration of an FMCA in terms of its mpc. The results clearly show the gain in expressiveness due to the notion of semi-controllability, as well as the reduction of the number of configurations needed to compute the orchestration due to the introduction of a partial order of products of FMCA\@. This inspired us to consider semi-controllability also in MSCA and to develop a partial order of controllers for MSCA in this paper.}

In~\cite{BBL19}, we presented Timed Service Contract Automata (TSCA) as an extension of the FMCA from~\cite{JSCP} with real-time constraints. 
Formally, a configuration of a TSCA is a triple consisting of a recognised trace, a state, and a valuation of clocks.
The (finite) behaviour recognised by a TSCA are traces of alternating time and discrete transitions, i.e.\ in a given configuration either time progresses (a silent action in the languages recognised by TSCA) or a discrete step to a new configuration is performed.
Consequently, we defined an algorithm to compute the orchestration synthesis of TSCA\@. To respect the timing constraints, we used the notion of zones from timed games~\cite{AMPS98,CDFLL05}.
The resulting synthesis algorithm resembles a timed game, but it differs from classical timed game algorithms~\cite{AMPS98,CDFLL05,DLLNW10} by combining two separate games, viz.\ \emph{reachability} games (to ensure that marked states must be reachable) and \emph{safety} games (to ensure that forbidden states are never traversed).
A TSCA might be such that all bad configurations are unreachable (i.e.\ it is safe), while at the same time no final configuration is reachable (i.e.\ the resulting orchestration is empty).

\section{Conclusion}\label{sect:conclusion}

This paper presents our recent efforts, originally published in~\cite{BBP19}, concerning bridging the gap between the most permissive controller synthesis from Supervisory Control Theory with synthesis algorithms of orchestrations and choreographies for a formal model of service contracts called Modal Service Contract Automata. 
This includes a novel algorithm capable of synthesising a safe non-blocking composition of service contracts that is directly translatable into a choreographed formalism. 
A further contribution is an abstract synthesis algorithm that generalises the synthesis of the choreography, as well as that of the orchestration and that of the most permissive controller. 
This paper includes 
the proofs of all statements from~\cite{BBP19}. Furthermore, it contains a formal demonstration that the different synthesis algorithms are related through a notion of refinement, which allows us to formally prove that, under mild assumptions, the orchestration synthesis is an abstraction of the mpc synthesis. \modiff{Finally, the paper includes an extensive running example from the service domain that illustrates our contributions.}

The properties to be enforced in the algorithms presented in this paper are all invariants specified through either forbidden states or forbidden transitions. Future work is needed to investigate the abstract syntheses under other non-invariant properties.
Another avenue for future research is to investigate the different features of micro-services with respect to services, and to study what is needed to adapt the formalism of (timed/modal service) contract automata and our results to deal with micro-services.

\section*{Acknowledgments}
  \noindent We acknowledge useful comments from the reviewers and funding from the MIUR PRIN 2017FTXR7S project IT MaTTerS (Methods and Tools for Trustworthy Smart Systems).

\bibliographystyle{plainurl}
\bibliography{bib}

\iftoggle{APPENDIX}
{
	\appendix

\section{Proofs}\label{app:proofs}
We provide the proofs of Theorem~\ref{the:abstractorchestrationsynthesis} and Theorem~\ref{the:abstractchoreographysynthesis} only sketched in Section~\ref{sect:discussion}.  




\theabstractorchestrationsynthesis*

\theabstractchoreographysynthesis*



}
{}
\end{document}